\documentclass[12pt,a4paper]{amsart}

\usepackage{amsmath,amsthm,amsfonts,amssymb,srcltx,empheq}
\usepackage[dvipdfmx]{graphicx}
\usepackage[all]{xy}
\usepackage{color}
\usepackage{longtable}
\usepackage{hhline}
\usepackage{url}

\usepackage{here}

\usepackage{caption}
\def\IZ{\mathbb {Z}}

\def\IR{\mathbb {R}}
\def\IC{\mathbb {C}}

\def\Sym{\operatorname{Sym}}

\def\bu{\boldsymbol{u}}
\def\bv{\boldsymbol{v}}
\def\bm{\boldsymbol{m}}
\def\bd{\boldsymbol{d}}
\def\bk{\boldsymbol{k}}
\def\bL{\boldsymbol{L}}
\def\bI{\boldsymbol{I}}
\def\bS{\boldsymbol{S}}
\def\bx{\boldsymbol{x}}
\def\by{\boldsymbol{y}}

\def\x1{x^{(1)}}
\def\x2{x^{(2)}}
\def\xM1{x^{(M-1)}}
\def\xM{x^{(M)}}
\def\y1{y^{(1)}}
\def\y2{y^{(2)}}
\def\yM1{y^{(M-1)}}
\def\yM{y^{(M)}}

\def\ll {\left\lgroup}
\def\rr{\right\rgroup}

\theoremstyle{plain}
  \newtheorem{prob}{Problem}[section]

  \newtheorem{prop}[prob]{Proposition}
   
    \newtheorem{lemm}[prob]{Lemma}
    \newtheorem{defi}[prob]{Definition}
\theoremstyle{remark}
  \newtheorem{remark}[prob]{\bf Remark}

\newtheorem{exam}[prob]{\bf Example}

\usepackage[scale=0.82]{geometry}

\makeatletter
\def\Left#1#2\Right{\begingroup%
   \def\ts@r{\nulldelimiterspace=0pt \mathsurround=0pt}%
   \let\@hat=#1%
   \def\sht@im{#2}%
   \def\@t{{\mathchoice{\def\@fen{\displaystyle}\k@fel}%
          {\def\@fen{\textstyle}\k@fel}%
          {\def\@fen{\scriptstyle}\k@fel}%
          {\def\@fen{\scriptscriptstyle}\k@fel}}}%
   \def\g@rin{\ts@r\left\@hat\vphantom{\sht@im}\right.}%
   \def\k@fel{\setbox0=\hbox{$\@fen\g@rin$}\hbox{%
      $\@fen \kern.3875\wd0 \copy0 \kern-.3875\wd0%
      \llap{\copy0}\kern.3875\wd0$}}%
      \def\pt@h{\mathopen\@t}\pt@h\sht@im%
      \Right}%
\def\Right#1{\let\@hat=#1%
   \def\st@m{\mathclose\@t}%
   \st@m\endgroup}
\makeatother

\makeatletter
 \renewcommand{\theequation}{%
       \thesection.\arabic{equation}}
 \@addtoreset{equation}{section}
\makeatother

\makeatletter
\def\eqnarray{%
 \stepcounter{equation}%
 \let\@currentlabel=\theequation
 \global\@eqnswtrue
 \global\@eqcnt\z@
 \tabskip\@centering
 \let\\=\@eqncr
 $$\halign to \displaywidth\bgroup\@eqnsel\hskip\@centering
 $\displaystyle\tabskip\z@{##}$&\global\@eqcnt\@ne
 \hfil$\displaystyle{{}##{}}$\hfil
 &\global\@eqcnt\tw@$\displaystyle\tabskip\z@{##}$\hfil
 \tabskip\@centering&\llap{##}\tabskip\z@\cr}
\makeatother

\begin{document}

\title[]{Nested coordinate Bethe wavefunctions 
\\
from the Bethe/Gauge correspondence 
}

\author[]{Omar Foda and Masahide Manabe} 

\address{
School of Mathematics and Statistics, 
University of Melbourne, 
Royal Parade, Parkville, Victoria 3010, Australia
} 

\email{
omar.foda@unimelb.edu.au,
masahidemanabe@gmail.com} 

\keywords{
The Bethe/Gauge correspondence. 
2-dimensional gauged linear sigma models. 
The nested coordinate Bethe wavefunctions. 
XXX spin-chains.
} 

\dedicatory{To Professor Tetsuji Miwa on his 70th birthday}

\begin{abstract}
In \cite{Nekrasovtalk:1, Nekrasovtalk:2}, Nekrasov applied 
the Bethe/Gauge correspondence to derive the $\mathfrak{su}\, (2)$ 
XXX spin-chain coordinate Bethe wavefunction from the IR limit of 
a 2D $\mathcal{N}=(2, 2)$ supersymmetric $A_1$ quiver gauge theory 
with an orbifold-type codimension-2 defect. 
Later, Bullimore, Kim and Lukowski implemented Nekrasov's construction 
at the level of the UV $A_1$ quiver gauge theory, recovered his result, 
and obtained further 
extensions of the Bethe/Gauge correspondence \cite{Bullimore:2017lwu}. 
In this work, we extend the construction of the defect to $A_M$ quiver 
gauge theories to obtain the $\mathfrak{su} \, ( M + 1 )$ XXX spin-chain 
nested coordinate Bethe wavefunctions. The extension to XXZ spin-chain 
is straightforward. 
Further, we apply a Higgsing procedure to obtain more general $A_M$ 
quivers and the corresponding wavefunctions, and interpret this procedure  
(and the Hanany-Witten moves that it involves)
on the spin-chain side in terms of Izergin-Korepin-type specializations 
(and re-assignments) of the parameters of the coordinate Bethe wavefunctions.
\end{abstract} 

\maketitle

\section{Introduction}

In \cite{Nekrasov:2009uh,Nekrasov:2009ui}, Nekrasov and Shatashvili proposed 
the Bethe/Gauge correspondence between the on-shell Bethe eigenstates of XXX 
(resp. XXZ) spin-chain Hamiltonians and the vacuum states of 
2D $\mathcal{N}=(2,2)$ supersymmetric 
(resp. 3D $\mathcal{N}=2$) quiver gauge theories
\footnote{\, 
The on-shell Bethe states (the eigenstates of the spin-chain Hamiltonian) are such that 
the rapidity variables satisfy the Bethe equations, and the positions of 
the spin variables are summed over.
The off-shell Bethe states 
are such that 
the rapidity variables do not satisfy any conditions, and the positions of 
the spin variables are summed over \cite{gaudin.book, essler.book}.
}.

In \cite{Nekrasovtalk:1,Nekrasovtalk:2}, Nekrasov introduced 
an orbifold-type codimension-2 
defect in the IR limit of an $A_1$ quiver gauge theory on the gauge side of 
the correspondence,
and obtained the coordinate Bethe wavefunction in the $\mathfrak{su} \, (2)$ 
XXX spin-$\frac{1}{2}$ chain 
on the Bethe side \cite{gaudin.book, essler.book}
\footnote{\, 
We consider only spin-chains with periodic (possibly twisted) boundary 
conditions. The rapidity variables of the coordinate Bethe wavefunctions can 
be free, or can be set to satisfy the Bethe equations to enforce periodicity 
in the case of periodic boundary conditions. 
By definition, the positions of 
the spin variables in a coordinate Bethe wavefunction are fixed.
}. 

In \cite{Bullimore:2017lwu}, Bullimore, Kim and Lukowski implemented Nekrasov's
defect construction at the level of the UV A-twisted $\mathcal{N}=(2, 2)$ 
supersymmetric $A_1$ quiver gauge theories on $S^{\, 2}$, and used 
the localization formulae of \cite{Closset:2015rna,Benini:2015noa} to recover 
Nekrasov's result, amongst other results that further extend the Bethe/Gauge 
correspondence.

In this work, we extend the construction of \cite{Bullimore:2017lwu} to $A_M$ 
quiver gauge theories with orbifold-type codimension-2 defects to 
obtain $\mathfrak{su}(M+1)$ 
XXX spin-chain nested coordinate Bethe wavefunctions, with spin states in the fundamental 
representation.

Further, we apply a Higgsing procedure 
to obtain more general $A_M$ quiver gauge theories with 
codimension-2 defects, and their corresponding nested coordinate Bethe 
wavefunctions, and interpret these generalizations on the spin-chain 
side (and the Hanany-Witten moves that they involve) 
as Izergin-Korepin-type specializations (and re-assignments of the 
roles) of the parameters. 
While we focus on results in XXX spin-chains, we 
outline their straightforward extension to XXZ spin-chains.

\subsection{Outline of contents}

In Section \ref{sec:reform_loc_f}, we recall the localization formulae
of Closset, Cremonesi and Park, for 
2D $\mathcal{N}=(2, 2)$ quiver gauge theories on $S^{\,2}$ 
\cite{Closset:2015rna} (and that of Benini and Zaffaroni, for 
$\mathcal{N}=2$ quiver gauge theories on $S^{\, 2} \times S^{\, 1}$ 
\cite{Benini:2015noa}), used in \cite{Bullimore:2017lwu} to construct 
Nekrasov's orbifold defect in 2D $A_1$ quiver gauge theory.  
Generalizing the discussion in \cite{Bullimore:2017lwu} for $A_1$ 
quiver gauge theory, we formally introduce the equivariant characters 
and reconstruct the localization formulae in
\cite{Closset:2015rna,Benini:2015noa} from them.

In Section \ref{sec:bethe_gauge}, we recall the Bethe/Gauge correspondence 
\cite{Nekrasov:2009uh,Nekrasov:2009ui} 
for $A_M$ linear quiver gauge theories, and provide the equivariant characters 
for them.

In Section \ref{sec:off_Bethe_AM}, after briefly recalling the construction 
of orbifold defects in $A_1$ quiver gauge theory, we extend this construction 
to a simple $A_M$ linear quiver gauge theory and obtain 
the nested coordinate Bethe 
wavefunctions of the  $\mathfrak{su}(M+1)$ XXX spin-chain with spin states 
in the fundamental representation.

In Section \ref{sec:gen_AM}, using a Higgsing procedure, we generalize 
the orbifold construction of  
Section \ref{sec:off_Bethe_AM} to more general 
$A_M$ linear quiver gauge theories 
and find the corresponding partition functions (as specializations of 
coordinate Bethe wavefunctions) of the corresponding $\mathfrak{su}(M+1)$ 
vertex lattice models. 
In Section \ref{section_6}, we interpret the Higgsing procedure 
as an Izergin-Korepin-type specialization 
of the lattice parameters on the Bethe side of the Bethe/Gauge correspondence, 
and interpret the Hanany-Witten moves that are involved in the Higgsing on the 
gauge side as a re-assignment of the lattice parameters on the Bethe side, 
and in Section \ref{section_7}, we include remarks.

In Appendix \ref{app:A1PDWPF}, we discuss the partial domain wall partition 
functions (DWPFs) in the rational $\mathfrak{su}(2)$ (six-)vertex model and 
prove Proposition \ref{prop:det_su2_pDW}, and in Appendix \ref{app:AM_v_model}, 
we define the partition function of the $\mathfrak{su}(M+1)$ vertex model 
which corresponds to the orbifold defect in the $A_M$ quiver constructed 
in the present work.

\section{Rewriting the localization formulae}
\label{sec:reform_loc_f}

\textit{We review the localization formulae of the partition 
functions of the A-twisted gauged linear sigma models (GLSMs) 
\cite{Witten:1993yc,Morrison:1994fr}, which are 
2D A-twisted $\mathcal{N}=(2,2)$ supersymmetric gauge theories 
on the $\Omega$-deformed $S_{\hbar}^{\, 2}$ \cite{Closset:2015rna} 
(and 3D twisted $\mathcal{N}=2$ gauge theories 
on $S_{\hbar}^{\, 2} \times S^{\, 1}$ \cite{Benini:2015noa}), 
where $\hbar$ is the $\Omega$-deformation parameter. 
Following \cite{Bullimore:2017lwu}, we provide the localization 
formulae in terms of equivariant characters which are more fundamental 
objects than the partition functions.
The equivariant characters are used to construct orbifold defects in 
Sections \ref{sec:off_Bethe_AM} and \ref{sec:gen_AM}.}
\footnote{\, 
In the present work, we use
\textit{\lq equivariant character\rq}, 
as well as the notation  
$\hbar$ for the $\Omega$-deformation parameter, and
$\gamma$ for a twisted mass parameter.
In \cite{Bullimore:2017lwu}, 
Bullimore \textit{et al.} use 
\textit{\lq equivariant index\rq} instead of 
\textit{\lq equivariant character\rq}, 
as well as the notation 
$\epsilon$ for the $\Omega$-deformation parameter, and  
$\hbar$ for the twisted mass parameter.}

\subsection{Equivariant characters}

Consider the 2D $\mathcal{N}=(2,2)$ (or 3D $\mathcal{N}=2$) supersymmetric 
gauge theory consisting of a vector multiplet $V$ in a Lie algebra $\mathfrak{g}$, 
of gauge group $G$, and $L$ chiral matter multiplets 
$\Phi_{\mathfrak{R}_i}^{r_i}$, with representations $\mathfrak{R}_i$ 
of $\mathfrak{g}$, $U(1)$ vector $R$-charges $r_i$ and twisted masses 
$\lambda_i$, $i = 1, \ldots, L$. In the present work, $r_i=0$ 
or $2$. The vector multiplet $V$ contains scalars 
$\bu=\{u_1,\ldots,u_{\mathrm{rk}(\mathfrak{g})}\}$, which take values 
in $\mathfrak{h} \otimes_{\IR}{\IC}$ and parametrize the Coulomb 
branch of the gauge theory, where $\mathrm{rk}(\mathfrak{g})$ is 
the rank of $\mathfrak{g}$ and $\mathfrak{h}$ is the Cartan subalgebra 
of $\mathfrak{g}$.

\begin{defi}
\label{def:eq_character}
The equivariant characters of the vector multiplet $V$ and the chiral matter 
multiplet $\Phi_{\mathfrak{R}}^{r}$ are
\begin{equation}
\chi_{\hbar}^{V}(\bu)
:= - \sum_{\alpha \in \Delta_+} 
\frac{\mathrm{e}^{\, \alpha(\bu)} + \mathrm{e}^{- \alpha(\bu)}}
{1 - \mathrm{e}^{\, \hbar}},
\qquad
\chi_{\hbar}^{\Phi_{\, \mathfrak{R}}^{\, r}} (\bu)
:= \sum_{\rho \in \mathfrak{R}} 
\frac{\mathrm{e}^{\, \rho(\bu) + \, \lambda + \, \hbar\, \delta_{\, r, \, 2}}}
{1 - \mathrm{e}^{\, \hbar}},
\end{equation}
where $\Delta_+$ is the set of positive roots of $\mathfrak{g}$, while 
$\alpha(\bu)$ and $\rho(\bu)$ are the canonical pairings. The equivariant characters charged 
with GNO charges
$\bd=\{d_1, \ldots, d_{\, \mathrm{rk} (\mathfrak{g})}\}$, $d_a\in \pi_1(U(1)) 
\cong {\IZ}$, 
associated with the quantized magnetic fluxes of the $U(1)^{\, \mathrm{rk} (\mathfrak{g})}$ 
gauge fields on $S_{\hbar}^{\, 2}$, are
\begin{align}
\chi_{\hbar,\bd}^{P}(\bu):=
\chi_{+\hbar}^{P}(\bu|_{+})+
\chi_{-\hbar}^{P}(\bu|_{-}),
\label{eq_charge_ch}
\end{align}
where $P=V$ or $\Phi_{\mathfrak{R}}^{r}$, and
\begin{align}
u_a|_{\pm} = u_a \mp \frac{d_a}{2}\, \hbar
\end{align}
where the characters $\chi_{+\hbar}^{P}(\bu|_{+})$ 
and $\chi_{-\hbar}^{P}(\bu|_{-})$ are 
defined on the north pole and the south pole of $S_{\hbar}^{\, 2}$, respectively.
\end{defi}

\begin{defi}
The total equivariant character is
\begin{align}
\chi_{\hbar}^{\mathrm{total}}(\bu):=
\chi_{\hbar}^{V}(\bu)+
\sum_{i=1}^L \chi_{\hbar}^{\Phi_{\mathfrak{R}_i}^{r_i}}(\bu),
\end{align}
and the total equivariant character with charges $\bd$ is
\begin{align}
\chi_{\hbar,\bd}^{\mathrm{total}}(\bu):=
\chi_{+\hbar}^{\mathrm{total}}(\bu|_{+})+
\chi_{-\hbar}^{\mathrm{total}}(\bu|_{-})=
\chi_{\hbar,\bd}^{V}(\bu)+
\sum_{i=1}^L \chi_{\hbar,\bd}^{\Phi_{\mathfrak{R}_i}^{r_i}}(\bu)
\label{t_ind_d}
\end{align}
\end{defi}

\subsection{Partition functions}

Each character $\chi_{\hbar,\bd}^{P}(\bu)$ with charges 
$\bd$ can be expanded as
\begin{equation}
\label{w_v}
\chi_{\hbar,\bd}^{P}(\bu)=
\sum_{I} \mathrm{e}^{w_I(\bu;\bd;\hbar)}
-\sum_{J} \mathrm{e}^{v_J(\bu;\bd;\hbar)},
\end{equation}
where \eqref{w_v} defines $w_{\, I}$ and $v_{\, J}$, and we obtain 
the building blocks of the topologically twisted partition functions 
of 2D and 3D gauge theories by
\begin{align}
\widehat{\mathcal{Z}}_{\bd}^{P}(\bu;\hbar)=
\frac{\prod_J v_J(\bu;\bd;\hbar)}
{\prod_I w_I(\bu;\bd;\hbar)},
\label{rp_ch_Z}
\end{align}
and
\begin{align}
\widehat{\mathcal{Z}}_{\bd}^{\mathrm{K},P}(\bu;\hbar)=
\frac{\prod_J 2\sinh \ll v_J(\bu;\bd;\hbar)/2\rr }
{\prod_I 2\sinh \ll w_I(\bu;\bd;\hbar)/2\rr },
\label{rp_ch_Z_K}
\end{align}
respectively.

\begin{prop}
\label{prop:build_pf}
The partition functions \eqref{rp_ch_Z} and \eqref{rp_ch_Z_K}, 
which are obtained from Definition \ref{def:eq_character}, as
\begin{align}
\begin{split}
&
\widehat{\mathcal{Z}}_{\bd}^{V}(\bu;\hbar)=
(-1)^{\sum_{\alpha \in \Delta_+} \ll \alpha(\bd)+1\rr }\,
\prod_{\alpha \in \Delta_+} 
\ll \alpha(\bu)^{\, 2}-\frac{\alpha(\bd)^{\, 2}}{4}\, 
\hbar^{\, 2} \rr ,
\\
&
\widehat{\mathcal{Z}}_{\bd}^{\Phi_{\mathfrak{R}}^{r}}(\bu;\hbar)=
\prod_{\rho \in \mathfrak{R}} 
\frac{\hbar^{r-\rho(\bd)-1}}
{\ll \frac{\rho(\bu)+\lambda}{\hbar}
+\frac{r-\rho(\bd)}{2} \rr _{\rho(\bd)+1-r}},
\label{bb_oz}
\end{split}
\end{align}
and
\begin{align}
\begin{split}
&
\widehat{\mathcal{Z}}_{\bd}^{\mathrm{K},V}(\bu;\hbar)=
\ll -q^{-\frac12}\rr ^{\sum_{\alpha \in \Delta_+} \alpha(\bd)}\,
\prod_{\alpha \in \Delta_+} 
\ll 1-\bu^{-\alpha} q^{\frac12 \alpha(\bd)}\rr 
\ll 1-\bu^{\alpha} q^{\frac12 \alpha(\bd)}\rr ,
\\
&
\widehat{\mathcal{Z}}_{\bd}^{\mathrm{K},\Phi_{\mathfrak{R}}^{r}}(\bu;\hbar)=
\prod_{\rho \in \mathfrak{R}}
\frac{\ll \bu^{\rho} \Lambda \rr ^{\frac{\rho(\bd)
+1-r}{2}}}
{\ll \bu^{\rho} \Lambda q^{\frac{r-\rho(\bd)}{2}};q\rr _{\rho(\bd)+1-r}},
\label{bb_oz_k}
\end{split}
\end{align}
respectively, 
give the building blocks of 2D and 3D topologically twisted 
partition functions in \cite{Closset:2015rna,Benini:2015noa} 
(see Proposition \ref{prop:localization_formula}). 
Here $U_a= \mathrm{e}^{-u_a}$, $\Lambda= \mathrm{e}^{-\lambda}$, 
$q= \mathrm{e}^{-\hbar}$, 
$\bu^{\alpha}= \mathrm{e}^{-\alpha(\bu)}$, and
$\bu^{\rho}= \mathrm{e}^{-\rho(\bu)}$. 
\end{prop}

In the above proposition, the Pochhammer and $q$-Pochhammer symbols are, 
respectively, defined by
\begin{align}
\begin{split}
(x)_d=\frac{\Gamma(x+d)}{\Gamma(x)}=
\begin{cases}
\prod_{\ell=0}^{d} \ll x+\ell \rr \ \ 
&\textrm{if}\ d >0,
\\
1
&\textrm{if}\ d =0,
\\
\prod_{\ell=d}^{-1} \ll x+\ell \rr ^{-1}\ \ 
&\textrm{if}\ d <0,
\end{cases}
\end{split}
\end{align}
and
\begin{align}
\begin{split}
(x;q)_d=
\begin{cases}
\prod_{\ell=0}^{d} \ll 1-x q^{\ell} \rr \ \ 
&\textrm{if}\ d >0,
\\
1
&\textrm{if}\ d =0,
\\
\prod_{\ell=d}^{-1} \ll 1-x q^{\ell} \rr ^{-1}\ \ 
&\textrm{if}\ d < 0
\end{cases}
\end{split}
\end{align}

In the following, we assume that the gauge group $G$ contains central 
$U(1)^{\mathsf{c}}$ factors, and then 
one can deform the gauge theory by the associated Fayet-Iliopoulos 
(FI) parameters $\xi^{\, a} $, and theta angles $\theta^{\, a}$, $a=1,\ldots,\mathsf{c}$.

\begin{prop}[\cite{Closset:2015rna,Benini:2015noa}]
\label{prop:localization_formula}
Combining the complexified FI parameters 
$\tau^{\, a}  = \mathrm{i}\, \xi^{\, a}  + \frac{1}{2\pi}\, \theta^{\, a} $ with 
the building blocks in Proposition \ref{prop:build_pf} obtained 
from \eqref{t_ind_d}, up to sign factors, 
the partition function of the A-twisted GLSM on $S_{\hbar}^{\, 2}$ is given by
\begin{equation}
Z_{S_{\hbar}^{\, 2}} = \frac{1}{|\mathcal{W}|}
\sum_{\bd \in {\IZ}^{\mathrm{rk}(\mathfrak{g})}} 
\oint_{\Gamma} d^{\mathrm{rk}(\mathfrak{g})} u \, 
\widehat{\mathcal{Z}}_{\bd}^{\textrm{total}}(\bu;\hbar),
\label{A_ZS}
\end{equation} 
where
\begin{equation}
\widehat{\mathcal{Z}}_{\bd}^{\textrm{total}}(\bu;\hbar):=
\mathrm{e}^{2\pi \mathrm{i}\, \tau(\bd)}\,
\widehat{\mathcal{Z}}_{\bd}^{V}(\bu;\hbar)\,
\prod_{i=1}^L \widehat{\mathcal{Z}}_{\bd}^{\Phi_{\mathfrak{R}_i}^{r_i}}(\bu;\hbar)
\label{integrand_A_Z}
\end{equation}
Here $|\mathcal{W}|$ is the order of the Weyl group of $G$, and the pairing 
$\tau(\bd)=\sum_a \tau^{\, a}  d_a$ is defined by embedding ${\boldsymbol\tau}$ into 
$\mathfrak{h}^* \otimes_{\IR}{\IC}$. The contour integral along $\Gamma$ is 
given by 
the Jeffrey-Kirwan residue operation (JK contour integral) 
\cite{JeKi:1993,BrVer:1999,SzVe:2003} (see also \cite{Benini:2013xpa}), 
which  picks relevant poles of the integrand. Similarly, 
the correlation function of two codimension-2 defects, $\mathcal{O}_N(\bu)$ 
and $\mathcal{O}_S(\bu)$, inserted at the north pole and at the south pole 
of $S_{\hbar}^{\, 2}$, respectively, is given by
\begin{align}
\left<\mathcal{O}_N(\bu)\, \mathcal{O}_S(\bu)\right>_{S_{\hbar}^{\, 2}}=
\frac{1}{|\mathcal{W}|}
\sum_{\bd \in {\IZ}^{\mathrm{rk}(\mathfrak{g})}} 
\oint_{\Gamma} d^{\mathrm{rk}(\mathfrak{g})} u \, 
\widehat{\mathcal{Z}}_{\bd}^{\textrm{total}}(\bu;\hbar) \,
\mathcal{O}_N(\bu|_+)\, \mathcal{O}_S(\bu|_-)
\label{A_correl_S}
\end{align}
The topologically twisted partition function and correlation functions of 
the $\mathcal{N}=2$ gauge theory on $S_{\hbar}^{\, 2} \times S^{\, 1}$ 
are obtained, up to Chern-Simons factors, by replacing 
$\widehat{\mathcal{Z}}_{\bd}^{P}(\bu;\hbar)$ with 
$\widehat{\mathcal{Z}}_{\bd}^{\mathrm{K},P}(\bu;\hbar)$ in the above 
formulae.
\end{prop}

\section{The Bethe/Gauge correspondence}
\label{sec:bethe_gauge}

\begin{figure}[t]
\centering
\includegraphics[width=100mm]{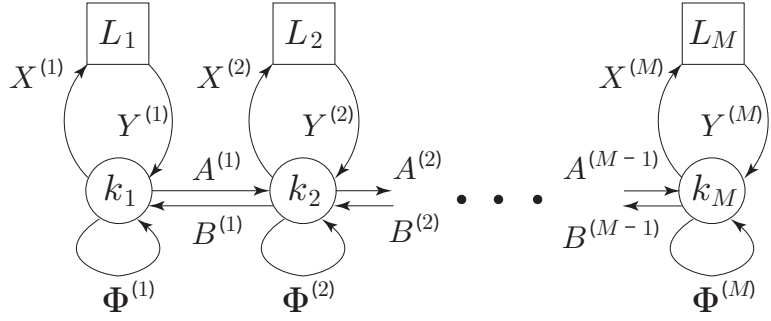}
\caption{$A_M$ linear quiver with the matter content in Table \ref{suM_matt}.}
\label{fig:am_quiver}
\end{figure}

\textit{We recall the basics of the Bethe/Gauge correspondence 
\cite{Nekrasov:2009uh,Nekrasov:2009ui}, 
between the supersymmetric vacua of 2D $\mathcal{N}=(2,2)$ 
(resp. 3D $\mathcal{N}=2$) $A_M$ linear quiver gauge 
theories as in Figure \ref{fig:am_quiver}, and 
the Bethe eigenfunctions of XXX (resp. XXZ) spin-chain Hamiltonians,  
with spins in the fundamental representation of $\mathfrak{su}(M+1)$.}

\medskip

\begin{table}[t]
\begin{center}
\begin{tabular}{|c|c|c|c|c|c|}
\hline
Field & $U(k_{p})$ & $U(k_{p+1})$ & $U(k_{q})$ & twisted mass& $U(1)_R$ 
\\ 
\hline
$X_i^{(p)}$ & $\mathbf{k}_p$ & $\mathbf{1}$ 
& $\mathbf{1}$ & $-m_i^{(p)}-\frac{\gamma}{2}$ & $0$ \\
$Y_i^{(p)}$ & $\overline{\mathbf{k}}_p$ 
& $\mathbf{1}$ & $\mathbf{1}$ & $m_i^{(p)}-\frac{\gamma}{2}$ 
& $0$ 
\\
$\Phi^{(p)}$ & {\bf adj} & $\mathbf{1}$ & $\mathbf{1}$ & $\gamma$ & $2$ 
\\ 
$A^{(p)}$ & $\mathbf{k}_p$ & $\overline{\mathbf{k}}_{p+1}$ & $\mathbf{1}$ & $-\frac{\gamma}{2}$ & $0$ \\
$B^{(p)}$ & $\overline{\mathbf{k}}_p$ & $\mathbf{k}_{p+1}$ & $\mathbf{1}$ & $-\frac{\gamma}{2}$ & $0$ \\
\hline
\end{tabular}
\caption{
The matter content of the $A_M$ quiver gauge theory in Figure \ref{fig:am_quiver} 
which describes an $\mathfrak{su}(M+1)$ spin-chain. 
For $X_i^{(p)}$, $Y_i^{(p)}$, $\Phi^{(p)}$,
$i=1, \ldots, L_{\, p}$ and $p=1, \ldots, M$. 
For $A^{(p)}$,  $B^{(p)}$,
$p=1, \ldots, M-1$.
Further, $q=1, \ldots, M$, with $q \neq p, p+1$, and  
$U(1)_R$ denotes the $U(1)$ vector $R$-charge.
}
\label{suM_matt}
\end{center}
\end{table}

The $A_M$ quiver gauge theory which corresponds to the $\mathfrak{su}(M+1)$ XXX spin-chain, 
with spins in the fundamental representation, is described by an A-twisted GLSM, 
on the $\Omega$-deformed $S_{\hbar}^{\, 2}$, with the gauge group 
$U(k_1)\times \cdots \times U(k_M)$, the matter content in Table \ref{suM_matt}, and the superpotential
\begin{align}
W = 
\sum_{p=1}^{M}
\sum_{i=1}^{L_{\, p}}
\sum_{a,b=1}^{k_{\, p}} 
X_{i,a}^{(p)} \, 
\Phi_{a,b}^{(p)} \, 
Y_{b,i}^{(p)}
+
\sum_{p=1}^{M-1}
\sum_{a,b=1}^{k_{\, p}}
\sum_{c=1}^{k_{p+1}} 
\Phi_{a,b}^{(p)} \, 
A_{b,c}^{(p)} \, 
B_{c,a}^{(p)}
+
\sum_{p=1}^{M-1}
\sum_{a=1}^{k_{\, p}}
\sum_{b,c=1}^{k_{p+1}}  
A_{a,b}^{(p)} \, 
\Phi_{b,c}^{(p+1)} \, 
B_{c,a}^{(p)}
\nonumber
\end{align}
It contains a set of vector multiplet scalars 
$\bu_{\bk}^M=\{\bu_{k_1}^{(1)},\ldots,\bu_{k_M}^{(M)}\}$,
with 
$\bu_{k_{\, p}}^{(p)}=\{u_1^{(p)},\ldots,u_{k_{\, p}}^{(p)}\}$, 
which parametrize the Coulomb branch, and twisted masses 
$\bm_{\bL}^M = \{ \bm_{L_1}^{(1)},$
$ \ldots,\bm_{L_M}^{(M)}\}$,  
with $\bm_{L_{\, p}}^{(p)}=\{m_1^{(p)},\ldots,m_{L_{\, p}}^{(p)}\}$ and $\gamma$, 
associated with the $U(L_1)\times \cdots \times U(L_M)\times U(1)$ flavor 
symmetry. In Table \ref{bg_dictionary}, we summarize the Bethe/Gauge dictionary 
\cite{Nekrasov:2009uh,Nekrasov:2009ui} which translates the gauge theory language 
to the spin-chain language. From Definition \ref{def:eq_character}, 
we obtain the  
equivariant characters for the $U(k_{\, p})$ vector multiplets $V^{(p)}$ 
and the chiral matter multiplets in Table \ref{suM_matt} as
\begin{align}
\begin{split}
&
\chi_{\hbar}^{V^{(p)}}
\ll \bu_{k_{\, p}}^{(p)} \rr 
=
- \sum_{a \ne b}^{k_{\, p}} 
\frac{
\mathrm{e}^{u_{a,b}^{(p)}}
}{
1-\mathrm{e}^{\hbar}
},
\qquad
\chi_{\hbar}^{\Phi^{(p)}}
\ll \bu_{k_{\, p}}^{(p)} \rr 
= 
\sum_{a,b=1}^{k_{\, p}} 
\frac{
\mathrm{e}^{u_{a,b}^{(p)} + \gamma + \hbar}
}{
1-\mathrm{e}^{\hbar}
},
\\
&
\chi_{\hbar}^{X_i^{(p)}}
\ll \bu_{k_{\, p}}^{(p)} \rr 
= \sum_{a=1}^{k_{\, p}} 
\frac{
\mathrm{e}^{u_a^{(p)} - m_i^{(p)} - \frac{\gamma}{2}}
}{
1-\mathrm{e}^{\hbar}
},
\qquad
\chi_{\hbar}^{Y_i^{(p)}}
\ll \bu_{k_{\, p}}^{(p)} \rr 
= \sum_{a=1}^{k_{\, p}} 
\frac{
\mathrm{e}^{-u_a^{(p)} + m_i^{(p)} - \frac{\gamma}{2}}
}{
1-\mathrm{e}^{\hbar}
},
\\
&
\chi_{\hbar}^{A^{(p)}}
\ll \bu_{k_{\, p}}^{(p)},\bu_{k_{p+1}}^{(p+1)} \rr 
= 
\sum_{a=1}^{k_{\, p}} 
\sum_{b=1}^{k_{p+1}} 
\frac{
\mathrm{e}^{u_a^{(p)} - u_b^{(p+1)} - \frac{\gamma}{2}}
}{
1-\mathrm{e}^{\hbar}
},
\\
&
\chi_{\hbar}^{B^{(p)}}
\ll \bu_{k_{\, p}}^{(p)},\bu_{k_{p+1}}^{(p+1)} \rr 
= 
\sum_{a=1}^{k_{\, p}} 
\sum_{b=1}^{k_{p+1}} 
\frac{
\mathrm{e}^{-u_a^{(p)} + u_b^{(p+1)} - \frac{\gamma}{2}}
}{
1-\mathrm{e}^{\hbar}
},
\label{chi_AM}
\end{split}
\end{align}
where $u_{a,b}^{(p)} = u_a^{(p)} - u_b^{(p)}$. 
\begin{table}[t]
\begin{center}
 \begin{tabular}{|c|c|c|}
\hline
& 2D/3D gauge with $A_M$ quiver (Figure \ref{fig:am_quiver}) 
& $\mathfrak{su}(M+1)$ XXX/XXZ spin-chain\\
\hline
$u_a^{(p)}$ & vector multiplet scalar & Bethe root \\
$m_i^{(p)}$ & twisted mass & inhomogeneity \\
$\gamma$ & twisted mass & coupling constant \\
$\mathfrak{q}_p$ & exponentiated FI parameter & 
periodic spin-chain boundary twist parameter \\
\hline
\eqref{eff_twist_sp_AM} & effective twisted superpotential & 
Yang-Yang function \\
\eqref{nBE_AM} & vacuum equation & nested Bethe equation \\
\hline
\end{tabular}
\caption{The Bethe/Gauge dictionary, where 
$a=1, \ldots, k_{\, p}$, 
$i=1, \ldots, L_{\, p}$, $p=1, \ldots, M$. 
Here $k_{\, p}$ is the number of vector multiplet scalars in $U (k_{\, p})$ 
on the gauge side/spin excitations on the Bethe side, 
and 
$L_{\, p}$ is the number of (anti-)fundamental flavors in $U(k_{\, p})$ 
on the gauge side/total number of spins on the Bethe side.}
\label{bg_dictionary}
\end{center}
\end{table}
Then, the building blocks of the $S^{\, 2}$ partition function 
in Proposition \ref{prop:build_pf} are obtained as
\begin{align}
\begin{split}
&
\widehat{\mathcal{Z}}_{\bd_{k_{\, p}}^{(p)}}^{V^{(p)}}
\ll \bu_{k_{\, p}}^{(p)}; \hbar \rr =
(-1)^{(k_{\, p}-1)\sum_{a=1}^{k_{\, p}} (d_a^{(p)}+\frac12)}
\prod_{a<b}^{k_{\, p}} \ll 
(u_{a,b}^{(p)})^{\, 2}-\frac{(d_{a,b}^{(p)})^{\, 2}}{4}\, \hbar^{\, 2}\rr ,
\\
&
\widehat{\mathcal{Z}}_{\bd_{k_{\, p}}^{(p)}}^{\Phi^{(p)}}
\ll \bu_{k_{\, p}}^{(p)}; \hbar \rr =
\prod_{a,b=1}^{k_{\, p}} \mathcal{Z}^{(d_{a,b}^{(p)}-2)}
\ll u_{a,b}^{(p)}+\gamma \rr,
\\
&
\widehat{\mathcal{Z}}_{\bd_{k_{\, p}}^{(p)}}^{X_i^{(p)}}
\ll \bu_{k_{\, p}}^{(p)}; \hbar \rr =
\prod_{a=1}^{k_{\, p}} \mathcal{Z}^{(d_a^{(p)})}
\ll u_a^{(p)}-m_i^{(p)}-\frac{\gamma}{2} \rr,
\\
&
\widehat{\mathcal{Z}}_{\bd_{k_{\, p}}^{(p)}}^{Y_i^{(p)}}
\ll \bu_{k_{\, p}}^{(p)}; \hbar \rr =
\prod_{a=1}^{k_{\, p}} \mathcal{Z}^{(-d_a^{(p)})}
\ll -u_a^{(p)}+m_i^{(p)}-\frac{\gamma}{2} \rr,
\\
&
\widehat{\mathcal{Z}}_{\bd_{k_{\, p}}^{(p)},\bd_{k_{p+1}}^{(p+1)}}^{A^{(p)}}
\ll \bu_{k_{\, p}}^{(p)},\bu_{k_{p+1}}^{(p+1)}; \hbar \rr =
\prod_{a=1}^{k_{\, p}} 
\prod_{b=1}^{k_{p+1}} 
\mathcal{Z}^{(d_a^{(p)}-d_b^{(p+1)})}
\ll u_a^{(p)}-u_b^{(p+1)}-\frac{\gamma}{2} \rr,
\\
&
\widehat{\mathcal{Z}}_{\bd_{k_{\, p}}^{(p)},\bd_{k_{p+1}}^{(p+1)}}^{B^{(p)}}
\ll \bu_{k_{\, p}}^{(p)},\bu_{k_{p+1}}^{(p+1)}; \hbar \rr =
\prod_{a=1}^{k_{\, p}} 
\prod_{b=1}^{k_{p+1}} 
\mathcal{Z}^{(-d_a^{(p)}+d_b^{(p+1)})}
\ll -u_a^{(p)}+u_b^{(p+1)}-\frac{\gamma}{2} \rr,
\label{Z_AM}
\end{split}
\end{align}
where $\bd_{\bk}^M=\{\bd_{k_1}^{(1)},\ldots,\bd_{k_M}^{(M)}\}$ 
with $\bd_{k_{\, p}}^{(p)}=\{d_1^{(p)},\ldots,d_{k_{\, p}}^{(p)}\}$, 
$d_a^{(p)}\in {\IZ}$, $p=1,\ldots,M$, are sets of GNO charges, 
$d_{a,b}^{(p)}=d_a^{(p)}-d_b^{(p)}$, and
\begin{align}
\mathcal{Z}^{(d)}(u)=
\begin{cases}
\prod_{\ell=0}^{d}\ll u-\frac{d}{2}\, \hbar + \ell\, \hbar \rr ^{-1}
\ \ &
\textrm{if}\ d \ge 0,
\\
\prod_{\ell=1}^{-d-1}\ll u+\frac{d}{2}\, \hbar + \ell\, \hbar \rr 
\ \ &
\textrm{if}\ d < 0
\end{cases}
\end{align}
Combining \eqref{Z_AM} with the complexified FI parameters $\tau^{\, p}$, $p=1,\ldots,M$,  
associated with the central $U(1)^M \subset U(k_1)\times \cdots \times U(k_M)$, the integrand 
\eqref{integrand_A_Z} of the $S^{\, 2}$ partition function is
\begin{multline}
\label{pf_AM}
\widehat{\mathcal{Z}}_{\bd_{\bk}^M}^{\textrm{total}}(\bu_{\bk}^M;\hbar)
=
\\ 
\prod_{p=1}^{M}
\mathfrak{q}_p^{\sum_{a=1}^{k_{\, p}} d_a^{(p)}}\,
\widehat{\mathcal{Z}}_{\bd_{k_{\, p}}^{(p)}}^{V^{(p)}}(\bu_{k_{\, p}}^{(p)};\hbar)\,
\widehat{\mathcal{Z}}_{\bd_{k_{\, p}}^{(p)}}^{\Phi^{(p)}}(\bu_{k_{\, p}}^{(p)};\hbar)\,
\times 
\prod_{i=1}^{L_{\, p}} \widehat{\mathcal{Z}}_{\bd_{k_{\, p}}^{(p)}}^{X_i^{(p)}}(\bu_{k_{\, p}}^{(p)};\hbar)\,
\widehat{\mathcal{Z}}_{\bd_{k_{\, p}}^{(p)}}^{Y_i^{(p)}}(\bu_{k_{\, p}}^{(p)};\hbar)
\\ 
\times 
\prod_{p=1}^{M-1} 
\widehat{\mathcal{Z}}_{\bd_{k_{\, p}}^{(p)},\bd_{k_{p+1}}^{(p+1)}}^{A^{(p)}}
(\bu_{k_{\, p}}^{(p)},\bu_{k_{p+1}}^{(p+1)};\hbar)\,
\widehat{\mathcal{Z}}_{\bd_{k_{\, p}}^{(p)},\bd_{k_{p+1}}^{(p+1)}}^{B^{(p)}}
(\bu_{k_{\, p}}^{(p)},\bu_{k_{p+1}}^{(p+1)};\hbar),
\end{multline}
where $\mathfrak{q}_p=(-1)^{L_{\, p}+k_{p-1}+k_{p}}\, \mathrm{e}^{2\pi \mathrm{i}\, \tau^{\, p}}$ 
are exponentiated FI parameters. 

In the limit $\hbar \to 0$, with positive FI parameters $\xi^{\, p} > 0$, and summing over 
the GNO charges $\bd_{\bk}^M$ with $d_a^{\, (p)} \ge 0$, the partition function \eqref{A_ZS} 
can be written in terms of a contour integral around the roots of the vacuum equations \cite{Closset:2015rna}
\begin{equation}
Z_{S^{\, 2}} =
\oint_{\Gamma_{\mathrm{eff}}} 
\prod_{p=1}^M \frac{d^{k_{\, p}} u^{(p)}}{k_{\, p}!}\, 
\frac{(-1)^{\frac{k_{\, p}}{2}(k_{\, p}-1)}}
{\prod_{a=1}^{k_{\, p}} \ll 1-\mathrm{e}^{2\pi \mathrm{i} \, 
\partial_{u_a^{(p)}}\mathcal{W}_{\mathrm{eff}}(\bu_{\bk}^M)}\rr }
\times \mathcal{Z}_{\mathbf{0}}(\bu_{\bk}^M),
\end{equation}
where
\begin{multline}
\mathcal{Z}_{\mathbf{0}}(\bu_{\bk}^M)
=
\prod_{p=1}^M
\frac{
\prod_{a<b}^{k_{\, p}}   \ll u_{a,b}^{(p)} \rr^{\, 2} \cdot 
\prod_{a,b=1}^{k_{\, p}} \ll u_{a,b}^{(p)}+\gamma \rr}
{\prod_{i=1}^{L_{\, p}} \prod_{a=1}^{k_{\, p}} 
\ll u_a^{(p)}  - m_i^{(p)}-\frac{\gamma}{2}\rr 
\ll -u_a^{(p)} + m_i^{(p)}-\frac{\gamma}{2}\rr 
}
\\
\times
\prod_{p=1}^{M-1}
\frac{1}{\prod_{a=1}^{k_{\, p}} \prod_{b=1}^{k_{p+1}} \ll u_a^{(p)}-u_b^{(p+1)}-\frac{\gamma}{2}\rr 
\ll -u_a^{(p)}+u_b^{(p+1)}-\frac{\gamma}{2}\rr}
\end{multline}
Here the effective twisted superpotential \cite{Nekrasov:2009uh} (in 
the denominator of the integrand),
\begin{align}
\mathcal{W}_{\mathrm{eff}}(\bu_{\bk}^M)=
\mathcal{W}_{\mathrm{cl}}(\bu_{\bk}^M)+
\mathcal{W}_{\mathrm{vec}}(\bu_{\bk}^M)+
\mathcal{W}_{\mathrm{matt}}(\bu_{\bk}^M),
\label{eff_twist_sp_AM}
\end{align}
consists of
\begin{align}
\begin{split}
&
\mathcal{W}_{\mathrm{cl}}(\bu_{\bk}^M)=
\sum_{p=1}^M \tau^{\, p} \, \sum_{a=1}^{k_{\, p}} u_a^{(p)},
\\
&
\mathcal{W}_{\mathrm{vec}}(\bu_{\bk}^M)=
-\frac{1}{2}\sum_{p=1}^M \sum_{a<b}^{k_{\, p}} u_{a,b}^{(p)} = 
-\frac{1}{2}\sum_{p=1}^M \sum_{a=1}^{k_{\, p}} \ll k_{\, p}-2a+1\rr  u_{a}^{(p)},
\\
&
\mathcal{W}_{\mathrm{matt}}(\bu_{\bk}^M)=
-\frac{1}{2\pi \mathrm{i}}\Bigg[
\sum_{p=1}^M \sum_{i=1}^{L_{\, p}} \sum_{a=1}^{k_{\, p}}
\ll u_a^{(p)}-m_i^{(p)}-\frac{\gamma}{2}\rr  \ll \log \ll u_a^{(p)}-m_i^{(p)}-\frac{\gamma}{2}\rr -1\rr 
\\
&\qquad\qquad\qquad
+\sum_{p=1}^M \sum_{i=1}^{L_{\, p}} \sum_{a=1}^{k_{\, p}}
\ll -u_a^{(p)}+m_i^{(p)}-\frac{\gamma}{2}\rr  \ll \log \ll -u_a^{(p)}+m_i^{(p)}-\frac{\gamma}{2}\rr -1\rr 
\\
&\qquad\qquad\qquad
+\sum_{p=1}^M \sum_{a,b=1}^{k_{\, p}}
\ll u_{a,b}^{(p)}+\gamma\rr  \ll \log \ll u_{a,b}^{(p)}+\gamma\rr -1\rr 
\\
&\qquad\qquad\qquad
+\sum_{p=1}^{M-1} \sum_{a=1}^{k_{\, p}} \sum_{b=1}^{k_{p+1}}
\ll u_a^{(p)}-u_b^{(p+1)}-\frac{\gamma}{2}\rr  \ll \log \ll u_a^{(p)}-u_b^{(p+1)}-\frac{\gamma}{2}\rr -1\rr 
\\
&\qquad\qquad\qquad
+\sum_{p=1}^{M-1} \sum_{a=1}^{k_{\, p}} \sum_{b=1}^{k_{p+1}}
\ll -u_a^{(p)}+u_b^{(p+1)}-\frac{\gamma}{2}\rr  \ll \log \ll -u_a^{(p)}+u_b^{(p+1)}-\frac{\gamma}{2}\rr -1\rr 
\Bigg] 
\nonumber
\end{split}
\end{align}
The contour $\Gamma_{\mathrm{eff}}$ encloses the roots of 
vacuum equations 
$\mathrm{e}^{2\pi \mathrm{i}\, \partial_{u_a^{(p)}} \mathcal{W}_{\mathrm{eff}}(\bu_{\bk}^M)}=1$,
\begin{multline}
\prod_{i=1}^{L_{\, p}} \frac{u_a^{(p)}-m_i^{(p)}-\frac{\gamma}{2}}{u_a^{(p)}-m_i^{(p)}+\frac{\gamma}{2}}
=
\mathfrak{q}_p\, \prod_{b\ne a}^{k_{\, p}} \frac{u_{a,b}^{(p)}-\gamma}{u_{a,b}^{(p)}+\gamma}
\times
\prod_{b=1}^{k_{p-1}} \frac{u_a^{(p)}-u_b^{(p-1)}+\frac{\gamma}{2}}
{u_a^{(p)}-u_b^{(p-1)}-\frac{\gamma}{2}}
\times
\prod_{b=1}^{k_{p+1}} \frac{u_a^{(p)}-u_b^{(p+1)}+\frac{\gamma}{2}}
{u_a^{(p)}-u_b^{(p+1)}-\frac{\gamma}{2}},
\label{nBE_AM}
\end{multline}
where $a=1,\ldots,k_{\, p},\ p=1,\ldots,M$, $k_0 = k_{M+1} = 0$. 
These vacuum equations are the nested Bethe equations of the $\mathfrak{su}(M+1)$ XXX spin-chain 
(see the Bethe/Gauge dictionary in Table 
\ref{bg_dictionary}), which is the starting point of 
the Bethe/Gauge correspondence \cite{Nekrasov:2009uh,Nekrasov:2009ui}.

\begin{remark}
The 3D uplift of the above results for XXZ spin-chains is 
straightforward \cite{Nekrasov:2009uh} (see also \cite{Gaiotto:2013bwa}).
\end{remark}

\section{Nested coordinate Bethe wavefunctions from orbifold defects}
\label{sec:off_Bethe_AM}

\begin{figure}[t]
 \centering
  \includegraphics[width=100mm]{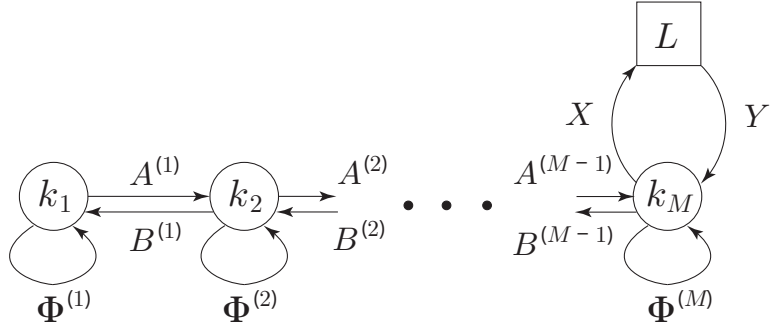}
\caption{The $A_M$ linear quiver that corresponds to Figure \ref{fig:am_quiver} 
with the ranks of the flavour groups 
$L_i=0$, $i=1,\ldots,M-1$, $L_M = L$, and the ranks of the gauge groups 
$k_1\le k_2 \le \ldots \le k_M \le L$.}
\label{fig:am_quiver_l}
\end{figure}

\textit{We review the construction of orbifold-type codimension-2 defects in  
the $A_1$ quiver gauge theory that corresponds to the $\mathfrak{su}(2)$ spin-chain \cite{Nekrasovtalk:1,Nekrasovtalk:2,Bullimore:2017lwu}, then extend that to the $A_M$ 
quiver gauge theory that corresponds to the $\mathfrak{su} (M + 1)$ spin-chain, 
described by the $A_M$ quiver in Figure \ref{fig:am_quiver_l}.}
\footnote{\,
For a suitable choice of the FI parameters, the GLSM described by the $A_M$ quiver 
in Figure \ref{fig:am_quiver_l} flows in the IR limit to a non-linear sigma model 
with the cotangent bundle of a partial flag variety as a target. This partial 
flag variety is defined by the set of subspaces
$
\{ \{0\} \subset {\IC}^{k_1} \subset {\IC}^{k_2} \subset \ldots 
\subset {\IC}^{k_M} \subset {\IC}^{L} \}
$,
in ${\IC}^{L}$. 
In the case of $k_{\, p} = p$, $M=L-1$, the variety is called 
the complete flag variety.}

\subsection{Orbifold defect for $A_1$ quiver}

\begin{table}[t]
\begin{center}
\begin{tabular}{|c|c|c|c|}
\hline
Field & $U(k)$ & twisted mass& $U(1)_R$ \\ \hline
$X_i$ & $\mathbf{k}$ & $-m_i-\frac{\gamma}{2}$ & $0$ \\
$Y_i$ & $\overline{\mathbf{k}}$ & $m_i-\frac{\gamma}{2}$ & $0$ \\
$\Phi$ & {\bf adj} & $\gamma$ & $2$ \\ \hline
\end{tabular}
\caption{The matter content of the $A_1$ quiver gauge theory that 
corresponds to the $\mathfrak{su}(2)$ spin-chain. 
Here $i=1, \ldots, L$, $k\le L$.}
\label{su2_matt}
\end{center}
\end{table}

Consider the A-twisted $U(k)$ GLSM on $S_{\hbar}^{\, 2}$, with matter 
content as in Table \ref{su2_matt}, and the superpotential 
$W=\sum_{i=1}^{L}\sum_{a,b=1}^{k} X_{i,a}\Phi_{a,b}Y_{b,i}$, $k \le L$. 
In this case, the equivariant characters \eqref{chi_AM} are given by
\begin{align}
\begin{split}
&
\chi_{\hbar}^{V}(\bu_k)
=- \sum_{a \ne b}^k \frac{\mathrm{e}^{u_{a,b}}}
{1-\mathrm{e}^{\hbar}},
\qquad
\chi_{\hbar}^{\Phi}(\bu_k)
= \sum_{a,b=1}^k \frac{\mathrm{e}^{u_{a,b} + \gamma + \hbar}}
{1-\mathrm{e}^{\hbar}},
\\
&
\chi_{\hbar}^{X_i}(\bu_k)
= \sum_{a=1}^k \frac{\mathrm{e}^{u_a - m_i - \frac{\gamma}{2}}}
{1-\mathrm{e}^{\hbar}},
\qquad
\chi_{\hbar}^{Y_i}(\bu_k)
= \sum_{a=1}^k \frac{\mathrm{e}^{-u_a + m_i - \frac{\gamma}{2}}}
{1-\mathrm{e}^{\hbar}}
\label{chi_A1}
\end{split}
\end{align}

Now, we recall the construction of orbifold defects for the $A_1$ quiver 
in \cite{Nekrasovtalk:1,Nekrasovtalk:2,Bullimore:2017lwu}. 
The orbifold defects inserted at the north (or south) pole 
of $S_{\hbar}^{\, 2}$ are characterized by a discrete holonomy 
$\omega^{\, n}$, $n = 0, 1, \ldots, L-1$, with $\omega^{\, L}=1$, 
associated with a ${\IZ}_L$ orbifold 
around the north (or south) pole, such that the gauge 
symmetry $U(k)$ and the flavor symmetry $U(L)$ are broken 
to a maximal torus. Firstly, for constructing such orbifold defects,
we change the parameters 
\begin{align}
m_i\ \to\ m_i + \ll i-1\rr \frac{\hbar}{L},\qquad
u_a|_{\pm}\ \to\ u_a|_{\pm} + \ll I_a-1\rr \frac{\hbar}{L},\qquad
\hbar\ \to\ \frac{\hbar}{L},
\label{orb_rep_su2}
\end{align}
in the total equivariant character $\chi_{+\hbar}^{\mathrm{total}}(\bu_{k}|_{+})$ 
or $\chi_{-\hbar}^{\mathrm{total}}(\bu_{k}|_{-})$ in \eqref{t_ind_d},
which is composed of the expressions in \eqref{chi_A1}, where
\begin{equation}
\bI_k=\{I_1,\ldots,I_k\} \subset 
\bI_L=\{1,\ldots,L\},\qquad I_a<I_{a+1},
\label{nest_a1}
\end{equation} 
is an ordered set that characterizes the orbifold defect. 
Next, taking the ${\IZ}_L$ invariant part under 
$\hbar \to \hbar + 2\pi \mathrm{i} n$, $n=0,1,\ldots,L-1$, 
of $\chi_{+\hbar}^{\mathrm{total}}(\bu_{k}|_{+})$ 
or $\chi_{-\hbar}^{\mathrm{total}}(\bu_{k}|_{-})$, we obtain
\begin{align}
\begin{split}
&
\chi_{\pm\hbar}^{\mathrm{total}}(\bu_{k}|_{\pm}) 
\\ 
& -
\sum_{a=1}^k \ll 
\sum_{i=1}^{I_a-1}\mathrm{e}^{\pm u_a|_{\pm} \mp m_i-\frac{\gamma}{2}}
+
\sum_{i=I_a+1}^{L}\mathrm{e}^{\mp u_a|_{\pm} \pm m_i-\frac{\gamma}{2}}
-\sum_{a<b}^k \ll \mathrm{e}^{\mp u_{a,b}|_{\pm}} + 
\mathrm{e}^{\pm u_{a,b}|_{\pm}+\gamma}
\rr 
\rr, 
\label{ch_def_a1}
\end{split}
\end{align}
which follows from the following lemma.

\begin{lemm}[\cite{Bullimore:2017lwu}]\label{lemma:ch_inv}
For any parameters $x$ and $\hbar$, 
and integers $I$ and $L$, performing the shift of parameters
\begin{align}
x\ \to\ x + I\, \frac{\hbar}{L},\qquad
\hbar\ \to\ \frac{\hbar}{L},\qquad -L+1 \le I \le L-1,
\end{align}
in $\mathrm{e}^{\pm x}/(1-\mathrm{e}^{\pm \hbar})$, 
leads to the ${\IZ}_L$ invariant part
\begin{align}
\frac{\mathrm{e}^{\pm x}}{1-\mathrm{e}^{\pm \hbar}} \ \
\longrightarrow\ \
\begin{cases}
\frac{\mathrm{e}^{\pm x}}{1-\mathrm{e}^{\pm \hbar}}\ \
&\textrm{if}\ \ -L+1 \le I \le 0,
\\
\frac{\mathrm{e}^{\pm x}}{1-\mathrm{e}^{\pm \hbar}} - \mathrm{e}^{\pm x}\ \
&\textrm{if}\ \ 1 \le I \le L-1
\end{cases}
\end{align}
\end{lemm}

\begin{proof}
By
$$
\frac{1}{1 - \mathrm{e}^{\, \frac{\hbar}{L}}}
=
\frac{
1 + 
\mathrm{e}^{\,    \frac{\hbar}{L}} + 
\mathrm{e}^{\, 2\,\frac{\hbar}{L}} + 
\cdots + 
\mathrm{e}^{\, (L-1) \, \frac{\hbar}{L}}
}{
1-\mathrm{e}^{\, \hbar}
},
$$
the lemma is proved.
\end{proof}

Symmetrizing in the variables $\bu_k$, the contribution of a defect inserted 
at the north (resp. south) pole of $S_{\hbar}^{\, 2}$ to the integrand of the 
JK contour integral 
\footnote{\,
In the sequel, we will also simply say 
\textit{\lq the (orbifold) defect\rq}, rather than 
\textit{\lq the contribution of the (orbifold) defect to the integrand of the JK contour integral representation of the gauge theory partition function\rq}.
}
is $\psi_{\bI_k}^{(L)}(\bu_k|_{+};\bm_L)$ 
(resp. $\psi_{\bI_k^{\vee}}^{(L)}(\bu_k|_{-};\bm_L^{\vee})$ 
with $I_a^{\vee}=L-I_a+1$ and $m_i^{\vee}=m_{L-i+1}$), where
\begin{align}
\psi_{\bI_k}^{(L)}(\bu_k;\bm_L)=
\mathop{\Sym}_{\bu_k}
\frac{\prod_{a=1}^k\ll \prod_{i=1}^{I_a-1}\ll u_a-m_i-\frac{\gamma}{2}\rr 
\cdot
\prod_{i=I_a+1}^{L}\ll -u_a+m_i-\frac{\gamma}{2}\rr \rr }
{\prod_{a<b}^k u_{b,a}\ll u_{a,b}+\gamma\rr }
\label{def_psi_a1}
\end{align}
Here $\mathop{\Sym}_{\, \bu_k}$ stands for the symmetrization of 
a function $f(\bu_k)$ in the variables $\bu_k$, 
$$
\mathop{\Sym}_{\bu_k}\, f(\bu_k)=
\sum_{\sigma \in \mathfrak{S}_k} f(u_{\sigma(1)},\ldots,u_{\sigma(k)}),
$$
and $\mathfrak{S}_k$ is the symmetric group of degree $k$.

\begin{prop}[\cite{Nekrasovtalk:1,Nekrasovtalk:2,Bullimore:2017lwu}]
\label{prop:def_a1}
Consider a normalization of the $A_1$ quiver orbifold defect as
\begin{align}
\label{defect_A1}
\widehat{\psi}_{\bI_k}^{(L)}(\bu_k;\bm_L)&=
(-1)^{kL+\sum_{a=1}^k I_a}\,
\prod_{a<b}^{k}\ll \gamma^{\, 2}-u_{a,b}^{\, 2}\rr \times 
\psi_{\bI_k}^{(L)}(\bu_k;\bm_L)
\\
&=
\mathop{\Sym}_{\bu_k}\, \omega_{\bI_k}^{(L)}(\bu_k;\bm_L),
\nonumber
\end{align}
where 
\begin{align}
\omega_{\bI_k}^{(L)}(\bu_k;\bm_L) =
\prod_{a=1}^k\ll \prod_{i=1}^{I_a-1}\ll u_a-m_i-\frac{\gamma}{2}\rr 
\times
\prod_{i=I_a+1}^{L}\ll u_a-m_i+\frac{\gamma}{2}\rr \rr 
\times 
\prod_{a<b}^k \frac{u_{a,b}-\gamma}{u_{a,b}}
\label{su2_sk}
\end{align}
Then, the defect $\widehat{\psi}_{\bI_k}^{(L)}(\bu_k;\bm_L)$ gives 
the coordinate Bethe wavefunction of $\mathfrak{su}(2)$ XXX 
spin-$\frac{1}{2}$ chain. 
\end{prop}

Note that the defects inserted at the north pole and the south pole of 
$S_{\hbar}^{\, 2}$ are given by
$$
\widehat{\psi}_{\bI_k}^{(L)} 
\ll \bu_k|_{+};\bm_L \rr \ \
\textrm{and}\ \
\widehat{\psi}_{\bI_k^{\vee}}^{(L)} 
\ll \bu_k|_{-};\bm_L^{\vee} \rr,
$$
respectively, where $I_a^{\vee}=L-I_a+1$ and $m_i^{\vee}=m_{L-i+1}$.

\begin{figure}[t]
 \centering
  \includegraphics[width=100mm]{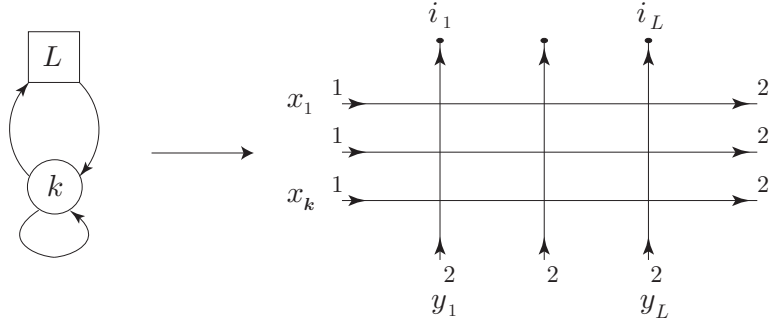}
\caption{
The $\mathfrak{su} (2)$ lattice configuration which describes 
the orbifold defect \eqref{defect_A1} for the $A_1$ quiver. 
Here $k \le L$, and the indices 
$i_{\ell} \in \{1,2\}$, $\ell=1,\ldots,L$, represent spin-states, or equivalently, \textit{colours}.
}
\label{fig:a1_lattice}
\end{figure}

\begin{remark}
Proposition \ref{prop:def_a1} implies that, setting $u_a = x_a + 1/2$, $m_i=y_i$, 
and $\gamma=1$, the defect $\widehat{\psi}_{\bI_k}^{(L)}(\bu_k;\bm_L)$ 
in \eqref{defect_A1} agrees with the partition function 
\eqref{am_lattice_pf} for the lattice 
configuration in Figure \ref{fig:a1_lattice} of the rational $\mathfrak{su}(2)$ 
six-vertex model, where the set $\bI_L$ labels the positions $\ell$ of colours 
$i_{\ell} \in \{1,2\}$, and the set $\bI_k$ labels the positions $k$ of 
colour $1$.
\end{remark}

\begin{remark}\label{rem:trig_A1_def}
We have derived $\widehat{\psi}_{\bI_k}^{(L)}(\bu_k;\bm_L)$ 
from the character in \eqref{ch_def_a1}, using \eqref{rp_ch_Z}. 
Using \eqref{rp_ch_Z_K} instead of \eqref{rp_ch_Z}, we obtain the coordinate Bethe 
wavefunction of the $\mathfrak{su}(2)$ XXZ spin-$\frac{1}{2}$ chain, 
and \eqref{su2_sk} is replaced by 
\begin{multline}
\omega_{\bI_k}^{\mathrm{K}, (L)}(\bu_k;\bm_L)=
\prod_{a=1}^k\ll \prod_{i=1}^{I_a-1}\left[u_a-m_i-\frac{\gamma}{2}\right]
\times 
\prod_{i=I_a+1}^{L}\left[u_a-m_i+\frac{\gamma}{2}\right]\rr 
\times 
\prod_{a<b}^k \frac{\left[u_{a,b}-\gamma\right]}{\left[u_{a,b}\right]},
\label{K_su2_sk}
\end{multline}
where $[x]=2\sinh (x/2)$.
\end{remark}

From the orbifold defect $\widehat{\psi}_{\bI_k}^{(L)}(\bu_k;\bm_L)$, 
one obtains the $\mathfrak{su}(2)$ six-vertex model partial domain wall partition
function (DWPF) \cite{Foda:2012yg}. 
In Appendix \ref{app:A1PDWPF}, we prove the following proposition as a corollary 
of Propositions \ref{prop:def_a1} and \ref{prop:det_su2}.

\begin{prop}[$A_1$ partial DWPF]\label{prop:det_su2_pDW}
Summing over the ordered set $\bI_k$ 
in the orbifold defect $\widehat{\psi}_{\bI_k}^{(L)}(\bu_k;\bm_L)$, 
we define a partition function
\begin{align}
\widehat{\mathcal{Z}}_{k}(\bu_k;\bm_L)=
\sum_{\bI_k\subset \{1,\ldots,L\}}
\widehat{\psi}_{\bI_k}^{(L)}(\bu_k;\bm_L) 
\label{wave_A1}
\end{align}
Then, the partition function $\widehat{\mathcal{Z}}_{k}(\bu_k;\bm_L)$ 
agrees with the $\mathfrak{su}(2)$ 
six-vertex model partial DWPF, 
which has the determinant expression \cite{Foda:2012yg},
\begin{align}
\begin{split}
\widehat{\mathcal{Z}}_k(\bu_k;\bm_L)&=
\frac{\prod_{a=1}^k\prod_{i=1}^L
\ll u_a-m_i-\frac{\gamma}{2}\rr 
\ll u_a-m_i+\frac{\gamma}{2}\rr }
{\prod_{a<b}^k (u_b-u_a) \cdot \prod_{i<j}^L (m_i-m_j)}
\\
&\qquad \times
\left|
\begin{array}{ccc}
\frac{1}{\ll u_1-m_1-\frac{\gamma}{2}\rr 
\ll u_1-m_1+\frac{\gamma}{2}\rr }
&\cdots&
\frac{1}{\ll u_1-m_L-\frac{\gamma}{2}\rr 
\ll u_1-m_L+\frac{\gamma}{2}\rr }\\
\vdots&{}&\vdots\\
\frac{1}{\ll u_k-m_1-\frac{\gamma}{2}\rr 
\ll u_k-m_1+\frac{\gamma}{2}\rr }
&\cdots&
\frac{1}{\ll u_k-m_L-\frac{\gamma}{2}\rr 
\ll u_k-m_L+\frac{\gamma}{2}\rr }\\
m_1^{L-k-1} & \cdots &m_L^{L-k-1}\\
\vdots&{}&\vdots\\
m_1^{0} & \cdots &m_L^{0}
\end{array}
\right|,
\end{split}
\end{align}
or equally \cite{Kostov:2012jr,Kostov:2012yq}, 
\begin{align}
\begin{split}
\widehat{\mathcal{Z}}_k(\bu_k;\bm_L)&=
\frac{\prod_{a=1}^k\prod_{i=1}^L
\ll u_a-m_i+\frac{\gamma}{2}\rr }
{\prod_{a<b}^k (u_b-u_a)}
\\
&\qquad \times
\det \ll 
u_a^{b-1}
\prod_{i=1}^L \frac{u_a-m_i-\frac{\gamma}{2}}{u_a-m_i+\frac{\gamma}{2}}
-(u_a-\gamma)^{b-1}\rr _{a,b=1,\ldots,k}
\end{split}
\end{align}
\end{prop}

\subsection{Orbifold defect for a simple $A_M$ quiver}
\label{subsec:AM_l_quiver}

We extend the above construction of the $A_1$ quiver orbifold defects to 
the simple $A_M$ linear quiver in Figure \ref{fig:am_quiver_l}, 
with the matter content in Table \ref{suM_matt}. 
The equivariant characters are given in \eqref{chi_AM}, with 
$L_{\, p}=0$, $p=1,\ldots,M-1$, and $k_1\le k_2 \le \ldots \le k_M \le L$. 
To construct orbifold defects which break the gauge symmetry 
$U(k_1)\times \cdots \times U(k_M)$ and the flavor symmetry $U(L)$ to 
a maximal torus, 
as a generalization of the change of parameters \eqref{orb_rep_su2}, 
we consider
\begin{align}
m_i\ \to\ m_i + \ll i-1\rr \frac{\hbar}{L},
\qquad
u_a^{(p)}\big|_{\pm}\ \to\ u_a^{(p)}\big|_{\pm} + \ll I_a^{(p)}-1\rr \frac{\hbar}{L},
\qquad
\hbar\ \to\ \frac{\hbar}{L},
\label{orb_rep_suM}
\end{align}
in the total equivariant character $\chi_{+\hbar}^{\mathrm{total}}(\bu_{\bk}^M|_{+})$ 
or $\chi_{-\hbar}^{\mathrm{total}}(\bu_{\bk}^M|_{-})$ in \eqref{t_ind_d}, 
where
\begin{equation} 
\bI_{k_1}^{(1)} \subset \bI_{k_2}^{(2)} \subset \cdots \subset 
\bI_{k_M}^{(M)} \subset \bI_L=\{1,\ldots,L\}, \quad
\bI_{k_{\, p}}^{(p)}=\big\{I_1^{(p)},\ldots,I_{k_{\, p}}^{(p)}\big\}, \quad
I_a^{(p)}<I_{a+1}^{(p)},
\label{nest_am}
\end{equation}
Taking the ${\IZ}_L$ invariant part of the total equivariant 
character under $\hbar \to \hbar + 2 \pi \mathrm{i} n$, $n=0,1,\ldots,L-1$, from Lemma \ref{lemma:ch_inv} one finds
\begin{align}
\begin{split}
&
\chi_{\pm\hbar}^{\mathrm{total}} \ll \bu_{\bk}^M|_{\pm} \rr 
\\
& -
\sum_{a=1}^{k_M} \ll 
\sum_{i=1}^{I_a^{(M)}-1}\mathrm{e}^{\pm u_a^{(M)}|_{\pm} \mp m_i-\frac{\gamma}{2}}
+
\sum_{i=I_a^{(M)}+1}^{L}\mathrm{e}^{\mp u_a^{(M)}|_{\pm} \pm m_i-\frac{\gamma}{2}}
\rr 
\\
& +
\sum_{p=1}^M \sum_{a<b}^{k_{\, p}} \ll \mathrm{e}^{\mp u_{a,b}^{(p)}|_{\pm}} + 
\mathrm{e}^{\pm u_{a,b}^{(p)}|_{\pm}+\gamma}
\rr
\\ 
& -
\sum_{p=1}^{M-1} \sum_{a=1}^{k_{\, p}} \ll 
\sum_{b=1}^{\widetilde{I}_a^{(p)}-1}
\mathrm{e}^{\pm u_a^{(p)}|_{\pm} \mp u_b^{(p+1)}|_{\pm}-\frac{\gamma}{2}}
+
\sum_{b=\widetilde{I}_a^{(p)}+1}^{k_{p+1}}
\mathrm{e}^{\mp u_a^{(p)}|_{\pm} \pm u_b^{(p+1)}|_{\pm}-\frac{\gamma}{2}}
\rr,
\end{split}
\end{align}
where the set $\widetilde{\bI}_a^{(p)}$ is defined by the map
\begin{align}
\bI_{k_{\, p}}^{(p)} \subset \bI_{k_{p+1}}^{(p+1)}
\quad \to \quad 
\widetilde{\bI}_{k_{\, p}}^{(p)} =
\big\{\widetilde{I}_1^{(p)},\ldots,\widetilde{I}_{k_{\, p}}^{(p)}\big\} \subset 
\{1,\ldots,k_{p+1}\}, 
\label{rearrange_I}
\end{align}
which can be explained as follows. 
$\bI_{k_{\, p}}^{(p)}$             is a subset in the set $\bI_{k_{p+1}}^{(p+1)}$, and
$\widetilde{\bI}_{k_{\, p}}^{(p)}$ is a subset in the set $\{1,\ldots,k_{p+1}\}$. 
Mapping the set $\bI_{k_{p+1}}^{(p+1)}$ to the set $\{1,\ldots,k_{p+1}\}$ using 
the map 
$I_{a}^{(p+1)} \mapsto a$, $a = 1, \ldots, p+1$,  
induces a map from the subset 
$\bI_{k_{\, p}}^{(p)}$ to the subset 
$\widetilde{\bI}_{k_{\, p}}^{(p)}$, which defines 
$\widetilde{\bI}_{k_{\, p}}^{(p)}$.  
By \eqref{rp_ch_Z}, after symmetrization in the 
vector multiplet scalars $\bu_{k_1}, \ldots, \bu_{k_M}$, one obtains the defect
\begin{multline}
\label{def_psi_aM}
\psi_{\bI_{\bk}^M}^{(L)}
\ll \bu_{\bk}^M; \bm_L \rr 
\\
=
\mathop{\Sym}_{\bu_{k_1}^{(1)},\ldots,\bu_{k_M}^{(M)}}
\prod_{p=1}^{M-1}
\frac{\prod_{a=1}^{k_{\, p}}
\ll \prod_{b=1}^{\widetilde{I}_a^{(p)}-1}
\ll u_a^{(p)}-u_b^{(p+1)}-\frac{\gamma}{2}\rr 
\cdot
\prod_{b=\widetilde{I}_a^{(p)}+1}^{k_{p+1}}
\ll -u_a^{(p)}+u_b^{(p+1)}-\frac{\gamma}{2}\rr \rr }
{\prod_{a<b}^{k_{\, p}} u_{b,a}^{(p)} \ll u_{a,b}^{(p)}+\gamma\rr}
\\
\times 
\frac{\prod_{a=1}^{k_M}
\ll \prod_{i=1}^{I_a^{(M)}-1}
\ll u_a^{(M)}-m_i-\frac{\gamma}{2}\rr 
\cdot
\prod_{i=I_a^{(M)}+1}^{L}
\ll -u_a^{(M)}+m_i-\frac{\gamma}{2}\rr \rr }
{\prod_{a<b}^{k_M} u_{b,a}^{(M)}\ll u_{a,b}^{(M)}+\gamma\rr},
\end{multline}
which generalizes $\psi_{\bI_k}^{(L)}(\bu_k;\bm_L)$ in \eqref{def_psi_a1}, 
where $\bI_{\bk}^M=\{\bI_{k_1}^{(1)},\ldots,\bI_{k_M}^{(M)}\}$, and 
$$
\mathop{\Sym}_{\bu_{k_1}^{(1)},\ldots,\bu_{k_M}^{(M)}}=
\mathop{\Sym}_{\bu_{k_1}^{(1)}}\cdots\mathop{\Sym}_{\bu_{k_M}^{(M)}}
$$
Using a normalization similar to that in \eqref{defect_A1}, we find the following proposition.

\begin{prop}\label{prop:def_aM}
The orbifold defect, for the simple $A_M$ quiver in Figure \ref{fig:am_quiver_l}, 
\begin{align}
\widehat{\psi}_{\bI_{\bk}^M}^{(L)}
\ll \bu_{\bk}^M;\bm_L \rr =
\mathop{\Sym}_{\bu_{k_1}^{(1)},\ldots,\bu_{k_M}^{(M)}} \prod_{p=1}^{M-1} 
\omega_{\widetilde{\bI}_{k_{\, p}}^{(p)}}^{(k_{p+1})}
\ll \bu_{k_{\, p}}^{(p)};\bu_{k_{p+1}}^{(p+1)} \rr  
\times
\omega_{\bI_{k_M}^{(M)}}^{(L)}
\ll \bu_{k_M}^{(M)};\bm_{L} \rr,
\label{defect_AM}
\end{align}
gives the nested coordinate Bethe wavefunction (see e.g. \cite{Mestyan:2017xyk}) 
in the $\mathfrak{su}(M+1)$ XXX spin-chain with spins in the fundamental
representation, where $\omega_{\bI_k}^{(L)}(\bu_k;\bm_L)$ is defined in 
\eqref{su2_sk}.
\end{prop}

\begin{remark}\label{rem:am_defect_lattice}
Proposition \ref{prop:def_aM} implies that, by the change of variables 
\eqref{mass_trans_lattice}, the orbifold defect \eqref{defect_AM} coincides 
with the partition function \eqref{am_lattice_pf} of a lattice configuration, 
with $L_i=0$, $i=1,\ldots,M-1$, and $L_M=L$, of the rational 
$\mathfrak{su} (M+1)$ XXX vertex model, where the set $\bI_L$ labels 
the positions $\ell$ of all colours $i_{\ell} \in \{1, \ldots, M+1\}$, 
the set $\bI_{k_1}^{(1)}$ labels the positions of colour $1$,
and the set 
$\bI_{k_{\, p}}^{\, (p)} \backslash \bI_{k_{p-1}}^{(p-1)}$, $p=2,\ldots,M$, 
labels the positions of colour $p$.
\end{remark}

\begin{remark}
In \cite{Shenfeld:2013}, the orbifold defect \eqref{defect_AM} was 
geometrically constructed as an element in \textit{a stable basis}
\cite{Maulik:2012wi} in the cotangent bundle of a partial flag 
variety (see also \cite{RTV:2015,Aganagic:2017gsx}).
\end{remark}

\begin{remark}
Similarly to Remark \ref{rem:trig_A1_def},
replacing $\omega_{\bI_k}^{(L)}(\bu_k;\bm_L)$ with 
$\omega_{\bI_k}^{\mathrm{K},(L)}(\bu_k;\bm_L)$ in the orbifold defect \eqref{defect_AM}, 
one obtains the trigonometric expression which coincides with coordinate Bethe wavefunction of the corresponding XXZ spin-chain.
\end{remark}

\begin{prop}\label{prop:det_suM_pDW}
Summing over the ordered sets $\bI_{\bk}^M$, 
with a fixed set $\bI_{k_M}^{(M)}$, 
in the defect \eqref{defect_AM}, we define a partition function
\begin{align}
\widehat{\mathcal{Z}}_{\bI_{k_M}^{(M)}}
\ll \bu_{\bk}^M;\bm_L \rr 
=
\sum_{\bI_{k_1}^{(1)}\subset \bI_{k_2}^{(2)}\subset \ldots \subset \bI_{k_M}^{(M)}}
\widehat{\psi}_{\bI_{\bk}^M}^{(L)}
\ll 
\bu_{\bk}^M; \bm_L 
\rr
\label{wave_AM}
\end{align}
Then, the partition function factorizes into the $A_1$ quiver orbifold defect
in \eqref{defect_A1} and the $\mathfrak{su} (2)$ six-vertex model partial 
DWPF \eqref{wave_A1}, 
\begin{align}
\widehat{\mathcal{Z}}_{\bI_{k_M}^{(M)}}
\ll 
\bu_{\bk}^M; \bm_L \rr 
=
\prod_{p=1}^{M-1} \widehat{\mathcal{Z}}_{k_{\, p}}
\ll \bu_{k_{\, p}}^{(p)};\bu_{k_{p+1}}^{(p+1)} \rr 
\times 
\widehat{\psi}_{\bI_{k_M}^{(M)}}^{(L)}
\ll \bu_{k_M}^{(M)};\bm_L \rr
\label{wave_AM_factor}
\end{align}
Further,
\begin{align}
\sum_{\bI_{k_M}^{(M)} \subset \{1,\ldots,L\}}
\widehat{\mathcal{Z}}_{\bI_{k_M}^{(M)}}
\ll \bu_{\bk}^M;\bm_L \rr 
=
\prod_{p=1}^{M-1} \widehat{\mathcal{Z}}_{k_{\, p}}
\ll \bu_{k_{\, p}}^{(p)};\bu_{k_{p+1}}^{(p+1)} \rr  
\times 
\widehat{\mathcal{Z}}_{k_M}
\ll \bu_{k_M}^{(M)};\bm_L \rr
\label{wave_AM_factor_p}
\end{align}
\end{prop}

\begin{proof}
From {\bf 3} in Proposition \ref{prop:cond_su2} and \eqref{limit_su2w}, 
the partition function $\widehat{\mathcal{Z}}_{k_1}(\bu_{k_1}^{(1)};\bu_{k_{2}}^{(2)})$ is 
a symmetric polynomial of $\bu_{k_{2}}^{(2)}$, and then,
\begin{equation}
\mathop{\Sym}_{\bu_{k_2}^{(2)}}\,
\widehat{\mathcal{Z}}_{k_1}
\ll \bu_{k_1}^{(1)}; \bu_{k_{2}}^{(2)} \rr \,
\omega_{\widetilde{\bI}_{k_2}^{(2)}}^{(k_{3})}
\ll \bu_{k_2}^{(2)}; \bu_{k_{3}}^{(3)} \rr 
=
\widehat{\mathcal{Z}}_{k_1}
\ll \bu_{k_1}^{(1)}; \bu_{k_{2}}^{(2)} \rr \,
\widehat{\psi}_{\widetilde{\bI}_{k_2}^{(2)}}^{(k_{3})}
\ll \bu_{k_2}^{(2)};\bu_{k_{3}}^{(3)} \rr 
\end{equation}
Repeating this symmetrization, one finds the factorizations 
\eqref{wave_AM_factor} and \eqref{wave_AM_factor_p}.
\end{proof}

\begin{remark}
In \cite{Foda:2013gb}, the factorization \eqref{wave_AM_factor_p} 
was shown for the $\mathfrak{su} (M+1)$ vertex model (see Remark
\ref{rem:am_defect_lattice}), and was traced to a property (called 
\textit{\lq colour-independence\rq}) of partition functions that 
satisfy certain boundary conditions.
\end{remark}

\section{Generalizations by Higgsing}
\label{sec:gen_AM}

\textit{In Section \ref{section_5_1}, we recall the Higgsing 
procedure in the $A_M$ quiver in Figure 
\ref{fig:am_quiver} without orbifold defects, 
in terms of the equivariant characters \eqref{chi_AM}.
In Sections \ref{section_5_2} and \ref{section_5_3}, by applying the Higgsing 
to the orbifold construction in Section \ref{sec:off_Bethe_AM}, 
we generalize the simple $A_M$ quiver orbifold defect \eqref{defect_AM} 
to more general $A_2$ then to $A_M$ quiver orbifold defects.
In Section \ref{section_5_4}, we study the dual of Higgsing
on the Bethe side of the Bethe/Gauge correspondence.
} 

\subsection{Higgsing}
\label{section_5_1}

\begin{table}[t]
\begin{center}
\begin{tabular}{|c||c c | c c c c c c c c|}
\hline
IIA & 0 & 1 & 2 & 3 & 4 & 5 & 6 & 7 & 8 & 9 
\\ 
\hline
NS5 & $-$ & $-$ & $-$ & $-$ & $-$ & $-$ & & & &
\\
D2 & $-$ & $-$ & & & & & $-$ & & & 
\\
D4 & $-$ & $-$ & & & & & & $-$ & $-$ & $-$ 
\\
\hline
\end{tabular}
\qquad
\begin{tabular}{|c||c c c | c c c c c c c|}
\hline
IIB & 0 & 1 & 2 & 3 & 4 & 5 & 6 & 7 & 8 & 9 
\\ 
\hline
NS5 & $-$ & $-$ & $-$ & $-$ & $-$ & $-$ & & & &
\\
D3 & $-$ & $-$ & $-$ & & & & $-$ & & & 
\\
D5 & $-$ & $-$ & $-$ & & & & & $-$ & $-$ & $-$ 
\\
\hline
\end{tabular}
\caption{On the left is 
a type-IIA brane configuration, and 
on the right is
the type-IIB brane configuration in \cite{Gaiotto:2013bwa} 
that corresponds to it by T-duality 
along the $x^2$-direction (see also Figure \ref{fig:aM_hw_brane}).}
\label{gk_hw_brane}
\end{center}
\end{table}

\begin{figure}[t]
 \centering
  \includegraphics[width=160mm]{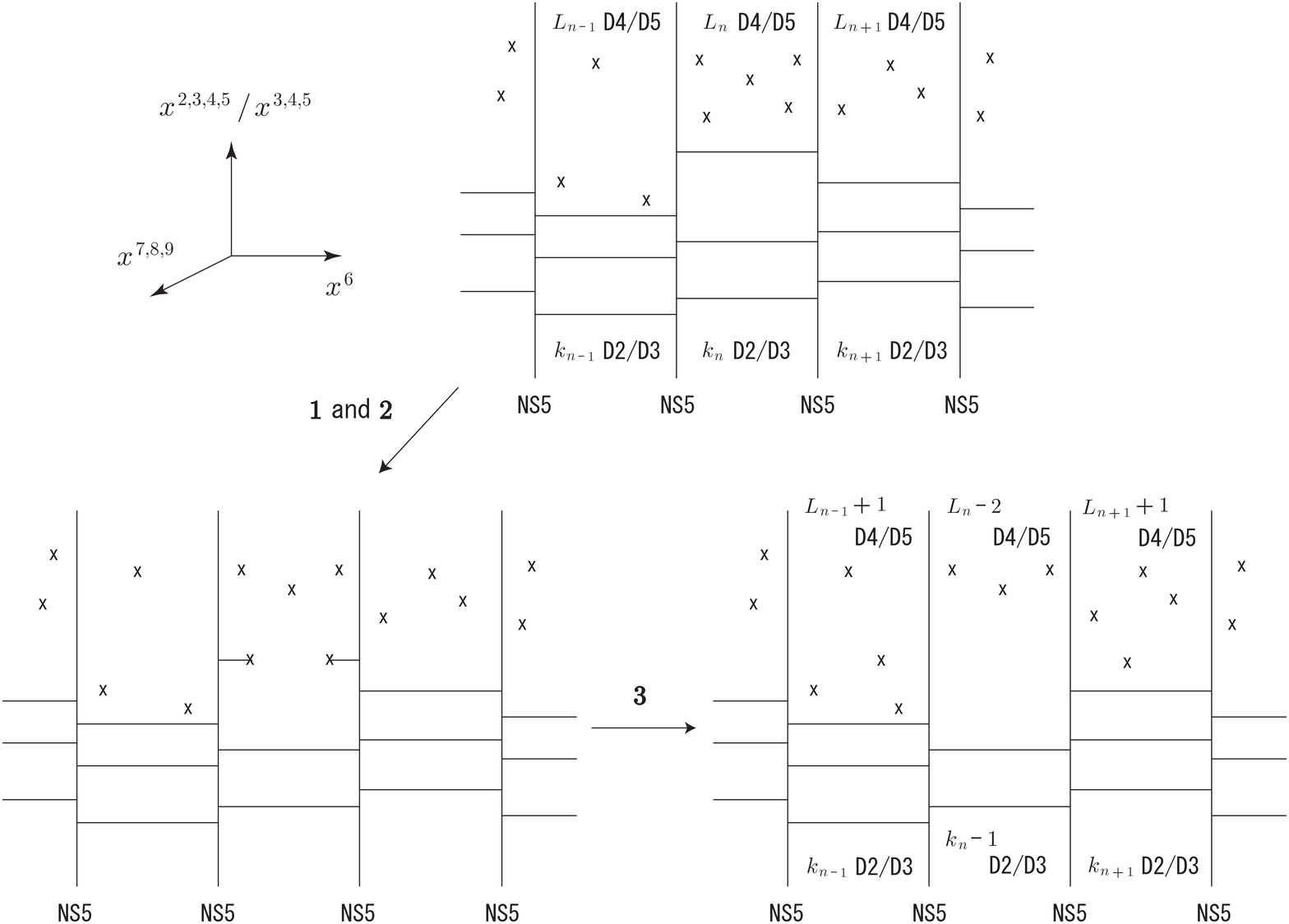}
\caption{The Higgsing procedure including Hanany-Witten moves for 
type-IIA/IIB brane configuration \cite{Gaiotto:2013bwa}, 
where the vertical (resp. horizontal) lines represent NS5 (resp. D2/D3) branes, 
and each x represents a D4/D5 brane. 
The corresponding quiver description is in Figure \ref{fig:aM_hw}.}
\label{fig:aM_hw_brane}
\end{figure}

In \cite{Gaiotto:2013bwa}, Gaiotto and Koroteev discussed a Higgsing 
procedure in terms of type-IIB brane realizations, in 
3D $\mathcal{N}=2$ quiver gauge theories 
that are Bethe/Gauge dual to XXZ spin-chains.
In Table \ref{gk_hw_brane} and Figure \ref{fig:aM_hw_brane}, 
we describe 
a type-IIA brane configuration, and 
a type-IIB brane configuration in \cite{Gaiotto:2013bwa} that 
corresponds to it by T-duality along the $x^2$-direction.
By introducing the twisted mass $\gamma$ in Table \ref{suM_matt}, 
which breaks half the supersymmetry, 
these configurations describe, respectively,
2D $\mathcal{N}=(2,2)$ quiver gauge theories on the $(x^0, x^1)$-directions,
and 
3D $\mathcal{N}=2$ quiver gauge theories on the $(x^0, x^1, x^2)$-directions. 
For the purposes of this work, the basic idea, as described in Figure \ref{fig:aM_hw_brane}, 
is
\begin{itemize}
\item[{\bf 1.}] to fine-tune the quiver data so that two D4/D5 branes are 
aligned at the same position in 
$(x^2, x^3, x^4, x^5)/(x^3, x^4, x^5)$-directions, 
and a segment of a D2/D3 brane stretches between them,  
\item[{\bf 2.}] the D2/D3 segment that stretches between the two aligned 
D4/D5 branes is taken to infinity in the $(x^7, x^8, x^9)$-directions, 
and finally, 
\item[{\bf 3.}] a sequence of Hanany-Witten moves of the two aligned D4/D5 
branes, which are across NS5 branes, are used to simplify 
the resulting quiver
\footnote{\,
The Hanany-Witten moves require that there is at most one D$p$ brane between 
an NS5 brane and a D$(p+2)$ brane. In the present work, $p = 2$ or $3$, and, 
as in Figure \ref{fig:aM_hw_brane}, there is indeed at most one D2/D3 brane 
between an NS5 brane and a D4/D5 brane.
The moves also describe the creation/annihilation of branes. 
In the present work, as in \cite{Gaiotto:2013bwa}, the Higgsing procedure, 
involves the annihilation of D2/D3 branes, as in Step \textbf{3}
of Figure \ref{fig:aM_hw_brane}. We thank A Hanany for discussions on 
this point.
}.
\end{itemize}

\begin{figure}[t]
 \centering
  \includegraphics[width=160mm]{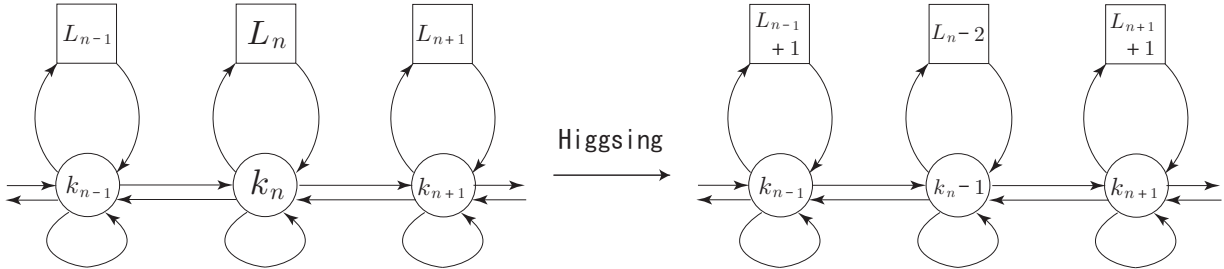}
\caption{Higgsing the $n$-th gauge node of the $A_M$ linear quiver.}
\label{fig:aM_hw}
\end{figure}

\subsubsection{Higgsing the $A_M$ quiver on the left hand side of Figure \ref{fig:aM_hw}}

Following \cite{Gaiotto:2013bwa}, the specialization of the parameters 
\begin{align}
u_{k_n}^{(n)}=m_{L_n-1}^{(n)}-\frac{\gamma}{2}
=m_{L_n}^{(n)}+\frac{\gamma}{2}=:\mu_{H},
\label{HW_move}
\end{align}
Higgses the $n$-th gauge node of the $A_M$ quiver on the left hand 
side of Figure \ref{fig:aM_hw}, and changes the quiver 
parameters as follows, 
\begin{align}
\begin{split}
&
k_{\, p}\ \to\ k_{\, p},\ (\textrm{for}\ p\ne n),\qquad
L_{\, p}\ \to\ L_{\, p},\ (\textrm{for}\ p\ne n, n \pm 1),
\\
&
k_n\ \to\ k_n-1,\qquad
L_n\ \to\ L_n-2,\qquad
L_{n \pm 1}\ \to\ L_{n \pm 1} + 1,
\label{HW_move_def}
\end{split}
\end{align}
as depicted in the transition from the left hand side to the right hand side 
of Figure \ref{fig:aM_hw} \cite{Hanany:1996ie}.

\subsubsection{Higgsing the equivariant characters \eqref{chi_AM}}

Applying \eqref{HW_move} to the characters, we obtain the Higgsed characters 
\begin{equation}
\chi_{H,\hbar}^{P}(\bu):=\chi_{\hbar}^{P}(\bu)\big|_{\eqref{HW_move}},
\end{equation} 
where the following characters remain unchanged
\begin{align}
\begin{split}
&
\chi_{H,\hbar}^{V^{(p)}} \ll \bu_{k_{\, p}}^{(p)} \rr =
\chi_{\hbar}^{V^{(p)}}   \ll \bu_{k_{\, p}}^{(p)} \rr ,\qquad
\chi_{H,\hbar}^{\Phi^{(p)}} \ll \bu_{k_{\, p}}^{(p)} \rr =
\chi_{\hbar}^{\Phi^{(p)}} \ll \bu_{k_{\, p}}^{(p)} \rr ,
\\
&
\chi_{H,\hbar}^{X_i^{(p)}} \ll \bu_{k_{\, p}}^{(p)} \rr =
\chi_{\hbar}^{X_i^{(p)}} \ll \bu_{k_{\, p}}^{(p)} \rr ,\qquad
\chi_{H,\hbar}^{Y_i^{(p)}} \ll \bu_{k_{\, p}}^{(p)} \rr =
\chi_{\hbar}^{Y_i^{(p)}} \ll \bu_{k_{\, p}}^{(p)} \rr ,\qquad 
\textrm{for}\ p \ne n,
\\
&
\chi_{H,\hbar}^{A^{(p')}}\ll \bu_{k_{p'}}^{(p')},\bu_{k_{p'+1}}^{(p'+1)} \rr =
\chi_{\hbar}^{A^{(p')}}\ll \bu_{k_{p'}}^{(p')},\bu_{k_{p'+1}}^{(p'+1)} \rr ,
\\
&
\chi_{H,\hbar}^{B^{(p')}}\ll \bu_{k_{p'}}^{(p')},\bu_{k_{p'+1}}^{(p'+1)} \rr =
\chi_{\hbar}^{B^{(p')}}\ll \bu_{k_{p'}}^{(p')},\bu_{k_{p'+1}}^{(p'+1)} \rr ,\qquad
\textrm{for}\ p' \ne n-1, n,
\label{hg_ch_AM_1}
\end{split}
\end{align}
while the following characters change
\begin{align}
\begin{split}
\label{hg_ch_AM_2}
&
\chi_{H,\hbar}^{V^{(n)}}\ll \bu_{k_n}^{(n)} \rr=
\chi_{\hbar}^{V^{(n)}}  \ll \bu_{k_n-1}^{(n)} \rr 
- \sum_{a=1}^{k_n-1} \frac{\mathrm{e}^{u_a^{(n)}-\mu_{H}}+\mathrm{e}^{-u_a^{(n)}+\mu_{H}}}{1-\mathrm{e}^{\hbar}},
\\
&
\chi_{H,\hbar}^{\Phi^{(n)}}\ll \bu_{k_n}^{(n)} \rr=
\chi_{\hbar}^{\Phi^{(n)}}\ll \bu_{k_n-1}^{(n)} \rr 
+\sum_{a=1}^{k_n-1} \frac{\mathrm{e}^{u_a^{(n)}-\mu_{H} + \gamma + \hbar}+\mathrm{e}^{-u_a^{(n)}+\mu_{H} + \gamma + \hbar}}
{1-\mathrm{e}^{\hbar}}
+\frac{\mathrm{e}^{\gamma + \hbar}}{1-\mathrm{e}^{\hbar}},
\\
&
\sum_{i=1}^{L_n} \chi_{H,\hbar}^{X_i^{(n)}}\ll \bu_{k_n}^{(n)} \rr=
\sum_{i=1}^{L_n-2} \chi_{\hbar}^{X_i^{(n)}}\ll \bu_{k_n-1}^{(n)} \rr 
+\sum_{a=1}^{k_n-1} \frac{\mathrm{e}^{u_a^{(n)} - \mu_H - \gamma} + 
\mathrm{e}^{u_a^{(n)} - \mu_H}}{1-\mathrm{e}^{\hbar}}
+\sum_{i=1}^{L_n} \frac{\mathrm{e}^{\mu_H - m_i^{(n)} - \frac{\gamma}{2}}}{1-\mathrm{e}^{\hbar}},
\\
&
\sum_{i=1}^{L_n} \chi_{H,\hbar}^{Y_i^{(n)}}\ll \bu_{k_n}^{(n)} \rr=
\sum_{i=1}^{L_n-2} \chi_{\hbar}^{Y_i^{(n)}}\ll \bu_{k_n-1}^{(n)} \rr 
+\sum_{a=1}^{k_n-1} \frac{\mathrm{e}^{-u_a^{(n)} + \mu_H} + 
\mathrm{e}^{-u_a^{(n)} + \mu_H - \gamma}}{1-\mathrm{e}^{\hbar}}
+\sum_{i=1}^{L_n} \frac{\mathrm{e}^{-\mu_H + m_i^{(n)} - \frac{\gamma}{2}}}{1-\mathrm{e}^{\hbar}},
\end{split}
\end{align}
and
\begin{align}
\begin{split}
\label{hg_ch_AM_3}
&
\chi_{H,\hbar}^{A^{(n-1)}}\ll \bu_{k_{n-1}}^{(n-1)},\bu_{k_n}^{(n)} \rr =
\chi_{\hbar}^{A^{(n-1)}} \ll \bu_{k_{n-1}}^{(n-1)},\bu_{k_n-1}^{(n)} \rr
+ \sum_{a=1}^{k_{n-1}} 
\frac{\mathrm{e}^{u_a^{(n-1)} - \mu_{H} - \frac{\gamma}{2}}}
{1-\mathrm{e}^{\hbar}},
\\
&
\chi_{H,\hbar}^{A^{(n)}} \ll \bu_{k_n}^{(n)},\bu_{k_{n+1}}^{(n+1)} \rr 
=
\chi_{\hbar}^{A^{(n)}} \ll \bu_{k_n-1}^{(n)},\bu_{k_{n+1}}^{(n+1)} \rr 
+ \sum_{a=1}^{k_{n+1}} 
\frac{\mathrm{e}^{\mu_{H} - u_a^{(n+1)} - \frac{\gamma}{2}}}
{1-\mathrm{e}^{\hbar}},
\\
&
\chi_{H,\hbar}^{B^{(n-1)}}\ll \bu_{k_{n-1}}^{(n-1)},\bu_{k_n}^{(n)} \rr =
\chi_{\hbar}^{B^{(n-1)}} \ll \bu_{k_{n-1}}^{(n-1)},\bu_{k_n-1}^{(n)} \rr + \sum_{a=1}^{k_{n-1}} 
\frac{\mathrm{e}^{-u_a^{(n-1)} + \mu_{H} - \frac{\gamma}{2}}}
{1-\mathrm{e}^{\hbar}},
\\
&
\chi_{H,\hbar}^{B^{(n)}} \ll \bu_{k_n}^{(n)},\bu_{k_{n+1}}^{(n+1)} \rr =
\chi_{\hbar}^{B^{(n)}} \ll \bu_{k_n-1}^{(n)},\bu_{k_{n+1}}^{(n+1)} \rr 
+ \sum_{a=1}^{k_{n+1}} 
\frac{\mathrm{e}^{-\mu_{H} + u_a^{(n+1)} - \frac{\gamma}{2}}}
{1-\mathrm{e}^{\hbar}}
\end{split}
\end{align}

The sum of the Higgsed characters in \eqref{hg_ch_AM_2} becomes
\begin{multline}
\chi_{H,\hbar}^{V^{(n)}}\ll \bu_{k_n}^{(n)} \rr+
\chi_{H,\hbar}^{\Phi^{(n)}}\ll \bu_{k_n}^{(n)} \rr+
\sum_{i=1}^{L_n} \ll  \chi_{H,\hbar}^{X_i^{(n)}}\ll \bu_{k_n}^{(n)} \rr +
\chi_{H,\hbar}^{Y_i^{(n)}}\ll \bu_{k_n}^{(n)} \rr \rr 
\\
=
\chi_{\hbar}^{V^{(n)}}\ll \bu_{k_n-1}^{(n)} \rr +
\chi_{\hbar}^{\Phi^{(n)}}\ll \bu_{k_n-1}^{(n)} \rr +
\sum_{i=1}^{L_n-2} \ll  \chi_{\hbar}^{X_i^{(n)}}\ll \bu_{k_n-1}^{(n)} \rr  +
\chi_{\hbar}^{Y_i^{(n)}}\ll \bu_{k_n-1}^{(n)} \rr  \rr 
\\
\quad
+\sum_{a=1}^{k_n-1} 
\ll 
\frac{\mathrm{e}^{-u_a^{(n)}+\mu_{H} - \gamma}}{1-\mathrm{e}^{\hbar}}
-\frac{\mathrm{e}^{u_a^{(n)}-\mu_{H} + \gamma}}{1-\mathrm{e}^{-\hbar}}
\rr 
+\sum_{a=1}^{k_n-1} 
\ll 
\frac{\mathrm{e}^{u_a^{(n)}-\mu_{H} - \gamma}}{1-\mathrm{e}^{\hbar}}
-\frac{\mathrm{e}^{-u_a^{(n)}+\mu_{H} + \gamma}}{1-\mathrm{e}^{-\hbar}}
\rr 
\\
+\frac{\mathrm{e}^{\gamma + \hbar}}{1-\mathrm{e}^{\hbar}}
+\sum_{i=1}^{L_n} \frac{\mathrm{e}^{\mu_H - m_i^{(n)} - \frac{\gamma}{2}}}{1-\mathrm{e}^{\hbar}}
+\sum_{i=1}^{L_n} \frac{\mathrm{e}^{-\mu_H + m_i^{(n)} - \frac{\gamma}{2}}}{1-\mathrm{e}^{\hbar}}
\end{multline}
By considering the partition function \eqref{rp_ch_Z} or \eqref{rp_ch_Z_K}, 
one finds that the contributions from the second line on the right hand side 
yield sign factors, while the third line does not depend on the variables 
$u_a^{(p)}$ and can be decoupled from the quiver gauge theory. 
One also finds that 
the extra factors of the first and third (resp. second and fourth) characters in 
\eqref{hg_ch_AM_3}, 
\begin{equation}
\sum_{a=1}^{k_{n-1}} 
\frac{\mathrm{e}^{u_a^{(n-1)} - \mu_{H} - \frac{\gamma}{2}}}
{1-\mathrm{e}^{\hbar}}
+ \sum_{a=1}^{k_{n-1}} 
\frac{\mathrm{e}^{-u_a^{(n-1)} + \mu_{H} - \frac{\gamma}{2}}}
{1-\mathrm{e}^{\hbar}},\quad
\textrm{resp.}\ 
\sum_{a=1}^{k_{n+1}} 
\frac{\mathrm{e}^{u_a^{(n+1)} -\mu_{H} - \frac{\gamma}{2}}}
{1-\mathrm{e}^{\hbar}}
+ \sum_{a=1}^{k_{n+1}} 
\frac{\mathrm{e}^{- u_a^{(n+1)} +\mu_{H} - \frac{\gamma}{2}}}
{1-\mathrm{e}^{\hbar}},
\end{equation}
agree with the contributions from extra (anti-)fundamental matter with 
mass parameter $\mu_H$ at the $(n-1)$-th (resp. $(n+1)$-th) gauge node. 
As a result, the transition \eqref{HW_move_def}, under the specialization 
\eqref{HW_move}, is confirmed.

\subsection{Orbifold defect for $A_2$ quiver}
\label{section_5_2}

In this subsection, we apply the Higgsing to the orbifold construction 
in Section \ref{sec:off_Bethe_AM} 
and generalize the simple $A_M$ quiver orbifold defects for 
$M=2$ constructed in Section \ref{subsec:AM_l_quiver}. 
The case of general $M$ will be discussed in Section \ref{section_5_3}. 

\begin{figure}[t]
 \centering
  \includegraphics[width=105mm]{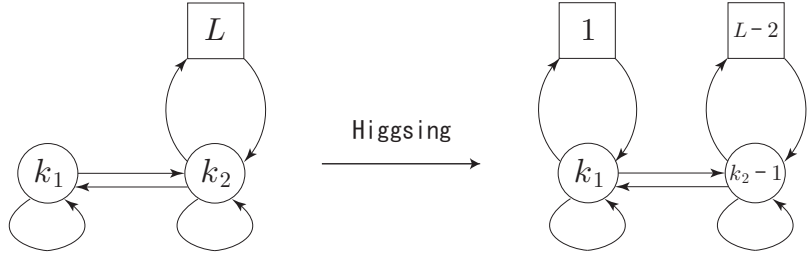}
\caption{Higgsing the $A_2$ linear quiver.}
\label{fig:a2_hw}
\end{figure}

\subsubsection{Higgsing the $A_2$ quiver gauge theory 
in Figure \ref{fig:a2_hw}}

We set the $A_2$ quiver data $k_1 \le k_2 \le L-1$, so 
that the specialization \eqref{HW_move} becomes
\begin{equation}
u_{k_2}^{(2)}=m_{L-1} - \frac{\gamma}{2}
=m_{L} + \frac{\gamma}{2}=:\mu_{H}
\end{equation}
As before, the twisted masses $m_{L-1}$ and $m_{L}$ decouple
after Higgsing, and compared with the discussion in Section 
\ref{subsec:AM_l_quiver}, one only needs to consider the second
factors on the right hand side in the Higgsed characters
\eqref{hg_ch_AM_3}.

\subsubsection{Construction of $A_2$ quiver orbifold defects}

Instead of the change of parameters \eqref{orb_rep_suM}, we consider
\begin{align}
\begin{split}
&
\mu_H\ \to\ \mu_H + \ll L-2\rr \frac{\hbar}{L-1},
\\
&
m_i\ \to\ m_i + \ll i-1\rr \frac{\hbar}{L-1},\qquad
i=1,\ldots,L-2,
\\
&
u_a^{(1)}\big|_{\pm}\ \to\ u_a^{(1)}\big|_{\pm} + \ll I_a^{(1)}-1\rr \frac{\hbar}{L-1},\qquad
a=1,\ldots,k_1,
\\
&
u_a^{(2)}\big|_{\pm}\ \to\ u_a^{(2)}\big|_{\pm} + \ll I_a^{(2)}-1\rr \frac{\hbar}{L-1},\qquad
a=1,\ldots,k_2-1,
\qquad
\hbar\ \to\ \frac{\hbar}{L-1},
\label{orb_rep_su3_hg1}
\end{split}
\end{align}
which is consistent with Higgsing, where
\begin{align}
\bI_{k_1}^{(1)} \subset \bI_{k_2-1}^{(2)} \uplus \{ L-1 \}
\subset \{1,\ldots,L-1\}, \quad
I_a^{(p)}<I_{a+1}^{(p)}, 
\end{align}
and the symbol $\uplus$ denotes the pairwise disjoint union. 
Then, as a generalization of \eqref{defect_AM} for $M=2$, we obtain 
an orbifold defect
\begin{multline}
\widehat{\psi}_{\bI_{k_1}^{(1)},\bI_{k_2-1}^{(2)}}^{(1|1,L-2)}
\ll \bu_{k_1}^{(1)},\bu_{k_2-1}^{(2)};\mu_H,\bm_{L-2} \rr = 
\\
\mathop{\Sym}_{\bu_{k_1}^{(1)},\bu_{k_2-1}^{(2)}} 
\omega_{\widetilde{\bI}_{k_1}^{(1)}}^{(k_2)}
\ll \bu_{k_1}^{(1)};\bu_{k_2-1}^{(2)},\mu_H \rr  \,
\omega_{\bI_{k_2-1}^{(2)}}^{(L-2)}
\ll \bu_{k_2-1}^{(2)};\bm_{L-2} \rr,
\label{defect_A3_hg}
\end{multline}
where $\omega_{\bI_k}^{(L)}(\bu_k;\bm_L)$ is defined in \eqref{su2_sk}, 
and $\widetilde{\bI}_{k_1}^{(1)}$ is defined by the map
\eqref{rearrange_I} for 
$\bI_{k_1}^{(1)} \subset \bI_{k_2-1}^{(2)} \uplus \{ L-1 \}$. Note that, 
instead of the change of parameters \eqref{orb_rep_su3_hg1}, by considering
\begin{align}
\begin{split}
&
\mu_H\ \to\ \mu_H,
\\
&
m_i\ \to\ m_i + i\,\frac{\hbar}{L-1},\qquad
i=1,\ldots,L-2,
\\
&
u_a^{(1)}\big|_{\pm}\ \to\ u_a^{(1)}\big|_{\pm} + \ll I_a^{(1)}-1\rr \frac{\hbar}{L-1},\qquad
a=1,\ldots,k_1,
\\
&
u_a^{(2)}\big|_{\pm}\ \to\ u_a^{(2)}\big|_{\pm} + 
\ll I_a^{(2)}-1 \rr
\, \frac{\hbar}{L-1},
\qquad
a=1,\ldots,k_2-1,
\qquad
\hbar\ \to\ \frac{\hbar}{L-1},
\label{orb_rep_su3_hg2}
\end{split}
\end{align}
with
\begin{align}
\bI_{k_1}^{(1)} \subset \{ 1 \} \uplus \bI_{k_2-1}^{(2)}
\subset \{1,\ldots,L-1\}, \quad
I_a^{(p)}<I_{a+1}^{(p)},
\end{align}
we obtain an another orbifold defect
\begin{multline}
\widehat{\psi}_{\bI_{k_1}^{(1)},\bI_{k_2-1}^{(2)}}^{(2|1,L-2)}
\ll \bu_{k_1}^{(1)},\bu_{k_2-1}^{(2)};\mu_H,\bm_{L-2} \rr 
=
\\ 
\mathop{\Sym}_{\bu_{k_1}^{(1)},\bu_{k_2-1}^{(2)}} 
\omega_{\widetilde{\bI}_{k_1}^{(1)}}^{(k_2)}
\ll \bu_{k_1}^{(1)};\mu_H,\bu_{k_2-1}^{(2)} \rr  \,
\omega_{\widetilde{\bI}_{k_2-1}^{(2)}}^{(L-2)}
\ll \bu_{k_2-1}^{(2)};\bm_{L-2} \rr,
\label{defect_A3_hg2}
\end{multline}
where $\widetilde{I}_a^{(2)}=I_a^{(2)}-1$.

\subsubsection{More general $A_2$ quiver orbifold defects}

One can apply the above Higgsing procedure repeatedly. 
Consider the $A_2$ quiver in Figure \ref{fig:am_quiver} with $M=2$, set 
$k_1 \le L_1 + k_2 \le L_1 + L_2$, and apply the change of parameters 
\begin{align}
\begin{split}
&
m_i^{(1)}\ \to\ m_i^{(1)} + \ll L_2+i-1\rr \frac{\hbar}{L_1+L_2},
\\
&
m_i^{(2)}\ \to\ m_i^{(2)} + \ll i-1\rr \frac{\hbar}{L_1+L_2},
\\
&
u_a^{(p)}\big|_{\pm}\ \to\ u_a^{(p)}\big|_{\pm} + \ll I_a^{(p)}-1\rr \frac{\hbar}{L_1+L_2},\qquad
\hbar\ \to\ \frac{\hbar}{L_1+L_2},
\label{orb_rep_su3_hg3}
\end{split}
\end{align}
to the equivariant characters \eqref{chi_AM}, with $M=2$, where
\begin{align}
\bI_{k_1}^{(1)} \subset \bI_{k_2}^{(2)} \uplus \{ L_2+1, \ldots, L_1+L_2 \}
\subset \{1,\ldots,L_1+L_2\}, \quad
I_a^{(p)}<I_{a+1}^{(p)}.
\end{align}
As a result, we find an orbifold defect for the $A_2$ quiver,
\begin{multline}
\widehat{\psi}_{\bI_{k_1}^{(1)},\bI_{k_2}^{(2)}}^{(1|L_1,L_2)}
\ll \bu_{k_1}^{(1)},\bu_{k_2}^{(2)};\bm_{L_1}^{(1)},\bm_{L_2}^{(2)} \rr =
\\ 
\mathop{\Sym}_{\bu_{k_1}^{(1)},\bu_{k_2}^{(2)}} 
\omega_{\widetilde{\bI}_{k_1}^{(1)}}^{(L_1+k_2)}
\ll \bu_{k_1}^{(1)};\bu_{k_2}^{(2)},\bm_{L_1}^{(1)} \rr  \,
\omega_{\bI_{k_2}^{(2)}}^{(L_2)}
\ll \bu_{k_2}^{(2)};\bm_{L_2}^{(2)} \rr,
\label{defect_A3_hg_gen}
\end{multline}
which generalizes the defect \eqref{defect_A3_hg}. Similarly, instead of 
\eqref{orb_rep_su3_hg3}, by considering
\begin{align}
\begin{split}
&
m_i^{(1)}\ \to\ m_i^{(1)} + \ll i-1\rr \frac{\hbar}{L_1+L_2},
\\
&
m_i^{(2)}\ \to\ m_i^{(2)} + \ll L_1+i-1\rr \frac{\hbar}{L_1+L_2},
\\
&
u_a^{(p)}\big|_{\pm}\ \to\ u_a^{(p)}\big|_{\pm} + \ll I_a^{(p)}-1\rr \frac{\hbar}{L_1+L_2},\qquad
\hbar\ \to\ \frac{\hbar}{L_1+L_2},
\label{orb_rep_su3_hg4}
\end{split}
\end{align}
with
\begin{align}
\bI_{k_1}^{(1)} \subset \{ 1, \ldots, L_1 \} \uplus \bI_{k_2}^{(2)}
\subset \{1,\ldots,L_1+L_2\}, \quad
I_a^{(p)}<I_{a+1}^{(p)},
\end{align}
we find an another orbifold defect for the $A_2$ quiver,
\begin{multline}
\widehat{\psi}_{\bI_{k_1}^{(1)},\bI_{k_2}^{(2)}}^{(2|L_1,L_2)}
\ll 
\bu_{k_1}^{(1)},\bu_{k_2}^{(2)};\bm_{L_1}^{(1)},\bm_{L_2}^{(2)}
\rr =
\\ 
\mathop{\Sym}_{\bu_{k_1}^{(1)},\bu_{k_2}^{(2)}} 
\omega_{\widetilde{\bI}_{k_1}^{(1)}}^{(L_1+k_2)}
\ll 
\bu_{k_1}^{(1)};\bm_{L_1}^{(1)},\bu_{k_2}^{(2)}
\rr  \,
\omega_{\widetilde{\bI}_{k_2}^{(2)}}^{(L_2)}
\ll 
\bu_{k_2}^{(2)};\bm_{L_2}^{(2)}
\rr,
\label{defect_A3_hg_gen2}
\end{multline}
which generalizes the defect \eqref{defect_A3_hg2}, 
where $\widetilde{I}_a^{(2)}=I_a^{(2)}-L_1$.

\subsection{Orbifold defect for $A_M$ quiver}
\label{section_5_3}

It is straightforward to generalize the above constructions of 
the orbifold defects \eqref{defect_A3_hg_gen} and \eqref{defect_A3_hg_gen2} 
to the $A_M$ quiver in Figure \ref{fig:am_quiver} (see Table 
\ref{suM_matt} for the matter content), where we set 
\begin{equation}
k_{\, p} \le L_{\, p} + k_{p+1}, \quad
k_M \le L_M, \quad p=1,\ldots,M-1
\end{equation}
The orbifold defects for the $A_2$ quiver are composed of 
the orbifold defect $\omega_{\bI_k}^{(L)}(\bu_k;\bm_L)$ in \eqref{su2_sk} 
for the $A_1$ quiver, and it is useful to define
\begin{align}
\begin{split}
&
\omega_{\bI_k}^{(1|\ell+L)}(\bu_k;\bv_{\ell},\bm_L):=
\omega_{\bI_k}^{(\ell+L)}(\bu_k;\bv_{\ell},\bm_L),
\\
&
\omega_{\bI_k}^{(2|\ell+L)}(\bu_k;\bv_{\ell},\bm_L):=
\omega_{\bI_k}^{(\ell+L)}(\bu_k;\bm_L,\bv_{\ell})
\label{omega_12}
\end{split}
\end{align}
Now, from the equivariant characters \eqref{chi_AM}, we find orbifold defects 
for the $A_M$ quiver, which generalize 
the simple $A_M$ quiver orbifold defect \eqref{defect_AM},
\begin{multline}
\widehat{\psi}_{\bI_{\bk}^M}^{(\bS|\bL)}
\ll \bu_{\bk}^M;\bm_{\bL}^M \rr =
\\
\mathop{\Sym}_{\bu_{k_1}^{(1)},\ldots,\bu_{k_M}^{(M)}} 
\prod_{p=1}^{M-1} 
\omega_{\widetilde{\bI}_{k_{\, p}}^{(p)}}^{(s_p|k_{p+1}+L_{\, p})}
\ll
\bu_{k_{\, p}}^{(p)};\bu_{k_{p+1}}^{(p+1)},\bm_{L_{\, p}}^{(p)}
\rr \times
\omega_{\widetilde{\bI}_{k_M}^{(M)}}^{(L_M)}
\ll 
\bu_{k_M}^{(M)};\bm_{L_M}^{(M)}
\rr,
\label{defect_AM_gen}
\end{multline}
which is characterized by the set 
$\bS=\{s_1,\ldots,s_{M-1}\}$, $s_p \in \{1,2\}$, 
and the sets $\bI_{\bk}^M$, with the inclusion relations
\begin{align}
\begin{split}
&
\bI_{k_{\, p}}^{(p)} \subset 
\widehat{\bI}_{k_{p+1}+L_{\, p}}^{(p+1)}:=
\begin{cases}
\bI_{k_{p+1}}^{(p+1)} \uplus 
\{L_{\, p}^{(\bS)}+L_{p+1}+1, \ldots, L_{\, p}^{(\bS)}+L_{\, p}+L_{p+1}\}
\ \ &\textrm{if}\ s_p =1,
\\
\{L_{\, p}^{(\bS)}+1, \ldots, L_{\, p}^{(\bS)}+L_{\, p}\} \uplus \bI_{k_{p+1}}^{(p+1)}
\ \ &\textrm{if}\ s_p =2,
\end{cases}
\\
& \hspace{1.6em} \subset 
\{L_{\, p}^{(\bS)}+1, \ldots, L_{\, p}^{(\bS)}+L_{\, p}+L_{p+1}\},\qquad
p=1,\ldots, M-1,
\\
&
\bI_{k_M}^{(M)} \subset \widehat{\bI}_{L_M}^{(M+1)}:=
\{L_M^{(\bS)}+1, \ldots, L_M^{(\bS)}+L_M\}
\end{split}
\end{align}
Here $L_{\, p}^{(\bS)}=\sum_{q=1}^{p-1} L_q \delta_{s_q,2}$, 
and the sets $\widetilde{\bI}_{\bk}^M$ are defined by 
the map \eqref{rearrange_I}, for the above inclusion relations. 
We claim that, by the reparametrizations \eqref{mass_trans_lattice}, 
the orbifold defect \eqref{defect_AM_gen} coincides  
with the partition function \eqref{am_lattice_pf} for a lattice 
configuration of the $\mathfrak{su} (M+1)$ vertex model, 
\begin{align}
 \widehat{\psi}_{\bI_{\bk}^M}^{(\bS|\bL)}
 \ll \bu_{\bk}^M;\bm_{\bL}^M \rr =   
 \widehat{\psi}_{\mathrm{L},\boldsymbol{i}_{\bL}^M}^{(\bS|\bL)}
 \ll \bx_{\bk}^M;\by_{\bL}^M \rr, 
\label{defect_lattice}
\end{align}
where 
$\widehat{\bI}_{k_{p+1}+L_{\, p}}^{(p+1)}\backslash \bI_{k_{p+1}}^{(p+1)}$,  
$p=1, \ldots, M-1$, and $\widehat{\bI}_{L_M}^{(M+1)}$ label the positions 
$\ell$ of colours $i_{\ell}^{(p)} \in \{1,\ldots,p+1\}$ and 
$i_{\ell}^{(M)} \in \{1,\ldots,M+1\}$, respectively, 
the set $\bI_{k_1}^{(1)}$ labels the positions of colour $1$, 
the set $\widehat{\bI}_{k_{\, p}+L_{p-1}}^{(p)} \backslash \bI_{k_{p-1}}^{(p-1)}$, 
$p=2, \ldots, M$, labels 
the positions of colour $p$, and 
the set $\widehat{\bI}_{L_{M}}^{(M+1)} \backslash \bI_{k_M}^{(M)}$ 
labels the positions of colour $M+1$.

\begin{figure}[t]
 \centering
  \includegraphics[width=135mm]{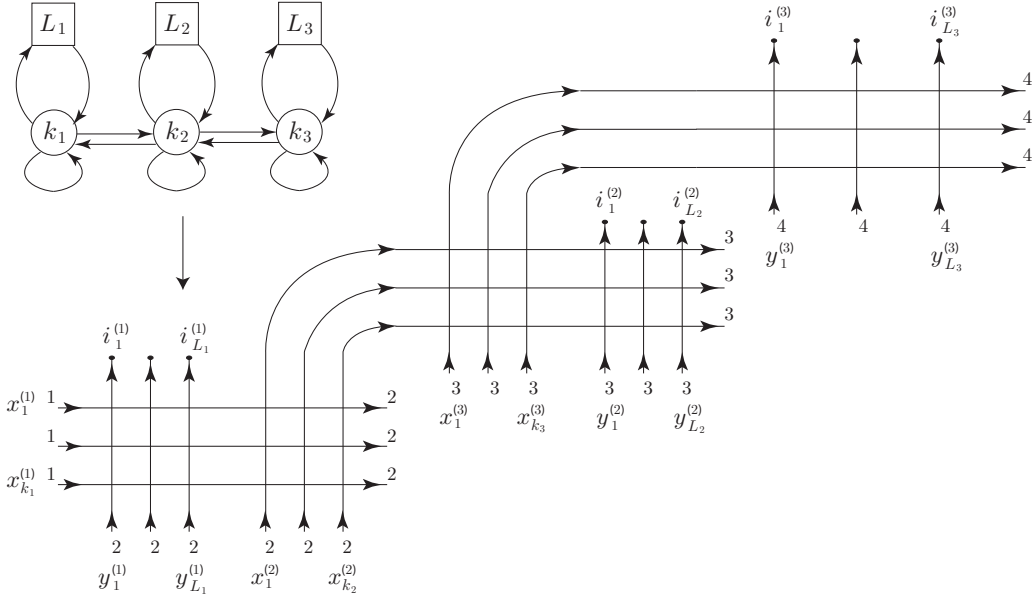}
\caption{An $\mathfrak{su} (4)$ lattice configuration with $s_1=2, s_2=1$ 
which describes the orbifold defect \eqref{defect_A3} for the $A_3$ quiver. 
Here $k_1 \le L_1+k_2$, $k_2 \le L_2+k_3$, $k_3 \le L_3$ and 
$i_{\ell}^{(p)} \in \{1,\ldots,p+1\}$, $\ell=1,\ldots,L_{\, p}$, $p=1,2,3$ 
represent colours.}
\label{fig:a3_lattice}
\end{figure}

\begin{exam}
Consider the $A_3$ quiver in Figure \ref{fig:am_quiver}, with $M=3$, and 
the change of parameters
\begin{align}
\begin{split}
&
m_i^{(1)}\ \to\ m_i^{(1)} + \ll i-1\rr \frac{\hbar}{L_1+L_2+L_3},
\\
&
m_i^{(2)}\ \to\ m_i^{(2)} + \ll L_1+L_3+i-1\rr \frac{\hbar}{L_1+L_2+L_3},
\\
&
m_i^{(3)}\ \to\ m_i^{(3)} + \ll L_1+i-1\rr \frac{\hbar}{L_1+L_2+L_3},
\\
&
u_a^{(p)}\big|_{\pm}\ \to\ u_a^{(p)}\big|_{\pm} + \ll I_a^{(p)}-1\rr \frac{\hbar}{L_1+L_2+L_3},\qquad
\hbar\ \to\ \frac{\hbar}{L_1+L_2+L_3},
\label{orb_rep_su4_hg}
\end{split}
\end{align}
with
\begin{align}
\begin{split}
\bI_{k_1}^{(1)} \subset \{ 1, \ldots, L_1 \} \uplus \bI_{k_2}^{(2)} 
&\subset \{ 1, \ldots, L_1 \} \uplus \bI_{k_3}^{(3)} 
\uplus \{ L_1+L_3+1, \ldots, L_1+L_2+L_3 \}
\\
&
\subset \{1,\ldots,L_1+L_2+L_3\},
\end{split}
\end{align}
where $I_a^{(p)}<I_{a+1}^{(p)}$. Our claim is that the orbifold defect
\begin{multline}
\label{defect_A3}
\widehat{\psi}_{\bI_{\bk}^3}^{(2,1|\bL)}
\ll \bu_{\bk}^3;\bm_{\bL}^3 \rr 
\\
=\mathop{\Sym}_{\bu_{k_1}^{(1)},\bu_{k_2}^{(2)},\bu_{k_3}^{(3)}} 
\omega_{\widetilde{\bI}_{k_1}^{(1)}}^{(2|k_{2}+L_1)}
\ll 
\bu_{k_1}^{(1)};\bu_{k_{2}}^{(2)},\bm_{L_1}^{(1)}
\rr \, 
\omega_{\widetilde{\bI}_{k_2}^{(2)}}^{(1|k_{3}+L_2)}
(\bu_{k_2}^{(2)};\bu_{k_{3}}^{(3)},\bm_{L_2}^{(2)})\, 
\omega_{\widetilde{\bI}_{k_3}^{(3)}}^{(L_3)}
\ll 
\bu_{k_3}^{(3)};\bm_{L_3}^{(3)}
\rr,
\end{multline}
coincides with the partition function \eqref{am_lattice_pf} for the $\mathfrak{su} (4)$ lattice 
configuration in Figure \ref{fig:a3_lattice} by the reparametrizations 
\eqref{mass_trans_lattice}. 
\end{exam}

\begin{remark}
In the $A_M$ quiver orbifold defects \eqref{defect_AM_gen}, 
by replacing the elementary blocks 
$\omega_{\bI_k}^{(s|\ell+L)}(\bu_k;\bv_{\ell},\bm_L)$ in \eqref{omega_12}
with
\begin{align}
\omega_{\bI_k}^{(\ell+L)}(\bu_k;\bm_{L^{(c)}},\bv_{\ell},\bm_{L-L^{(c)}}),
\end{align}
we obtain a further generalization of the orbifold defects, 
where $\bm_{L-L^{(c)}}=\{m_1,\ldots,m_{L-L^{(c)}}\}$ and 
$\bm_{L^{(c)}}=\{m_{L-L^{(c)}+1},\ldots,m_{L}\}$. 
By this generalization, our claim \eqref{defect_lattice} is obviously 
generalized to a claim for the lattice configurations which mix 
two lattice configurations in Figure \ref{fig:quiver_lattice_2} 
(and Figure \ref{fig:quiver_lattice_1}).
\end{remark}

\subsection{Higgsing on the Bethe side of the Bethe/Gauge correspondence}
\label{section_5_4}

As in Appendix \ref{app:AM_v_model}, the lattice configurations of the rational 
$\mathfrak{su} (M+1)$ vertex model consist of three types of vertices $a$, $b$ 
and $c$ in Figure \ref{fig:6v_abc} with vertex weights $a(x,y)=x-y+1$, 
$b(x,y)=x-y$ and $c(x,y)=1$. Consider the two types of lattice configurations 
in Figure \ref{fig:lattice_hw_1} consisting of $k+1$ horizontal-line 
variables $x$ and $x_a$, $a=1,\ldots,k$, and $L+2$ vertical-line variables 
$w_1$, $w_2$ and $y_{\ell}$, $\ell=1,\ldots,L$. 

\begin{figure}[H]
 \centering
  \includegraphics[width=105mm]{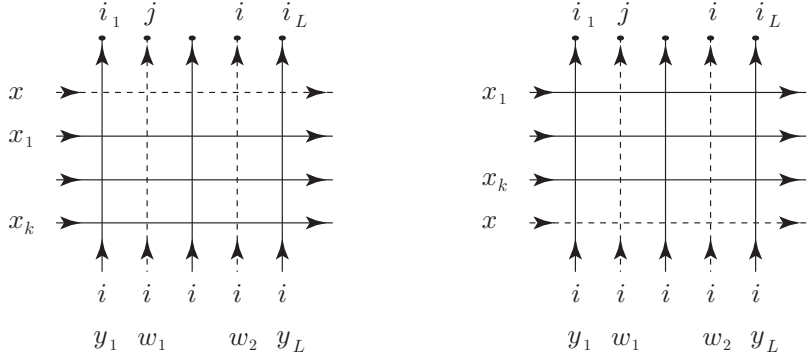}
\caption{Lattice configurations.}
\label{fig:lattice_hw_1}
\end{figure}

The bonds on the lower boundary are assigned the (fixed and same) colour $i$, and 
the bonds on the upper boundary are assigned (fixed but varying) colours 
$i$, $j$, $i_{\ell}$, $\ell=1, \ldots, L$, $j<i$. 
We take the colours $s$ coming from the left boundary to be less than $i$, 
i.e. $s<i$, while the colours that remain unspecified  
(such as those on the bonds on the right boundary) can take any value that 
is allowed by colour conservation. 
If the top bond of the left-most vertical line has colour $i$, then all 
vertices on this vertical line are of type-$b$ and their weights factor 
out trivially.

We now consider the Higgsing of the dashed lines by imposing the condition 
\eqref{HW_move}, where, by the change of variables \eqref{mass_trans_lattice} 
this condition yields
\begin{equation}
\textrm{A-type condition}:\ 
x_{k_n}^{(n)}=y_{L_{n-1}}^{(n)}-1 \quad \textrm{and} \quad
\textrm{B-type condition}:\ 
x_{k_n}^{(n)}=y_{L_n}^{(n)}
\label{ab_condition}
\end{equation}
In other words, 
$a ( x_{k_n}^{(n)},y_{L_{n-1}}^{(n)} ) =0$ and 
$b ( x_{k_n}^{(n)},y_{L_n}^{(n)} ) =0$, respectively. 
For the above two lattice configurations, let us impose the A-type condition 
$a(x,w_1)=0$ on the first dashed vertical line and the B-type condition 
$b(x,w_2)=0$ on the second dashed vertical line. 
For the Higgsing, we further consider $x \to \infty$ limit corresponding to the 
decoupling of the vector multiplet scalar. 
With the above conditions, the dominant lattice configuration should be one 
with the minimal number of $c$ vertices on the dashed lines. 

If the intersection of the horizontal dashed line and the vertical dashed line 
with top boundary colour $i$ is a $c$ vertex, then there are, at least, three 
(resp. two) $c$ vertices on the dashed lines for the left (resp. right) 
lattice configuration. To have the minimal number of $c$ vertices, this 
intersection should be an $a$ vertex by the B-type condition. 
In fact, in this case, the lattice configurations with minimal number of $c$ 
vertices are given in Figure \ref{fig:lattice_hw_2}.

\begin{figure}[H]
 \centering
  \includegraphics[width=105mm]{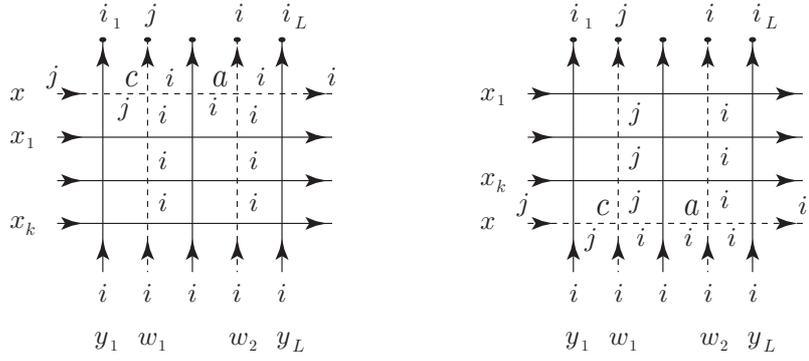}
\caption{Higgsed lattice configurations.}
\label{fig:lattice_hw_2}
\end{figure}

Here, on the dashed lines, there is exactly one $c$ vertex at the intersection 
of the horizontal dashed line and the vertical dashed line with top boundary 
colour $j$. The other vertices on the dashed lines are also uniquely determined, 
where the left and right side boundary colours of the horizontal dashed line are 
fixed as $j$ and $i$, respectively.
Note that without the second vertical dashed line with variable $w_2$, the 
lattice configuration on the dashed lines with a minimal number of $c$ vertex is 
not uniquely determined, and the Higgsing procedure for two vertical lines is 
inevitable. 

Now we consider that the horizontal dashed lines are connected with other lattice 
configuration in Figure \ref{fig:lattice_hw_3}.

\begin{figure}[H]
 \centering
  \includegraphics[width=160mm]{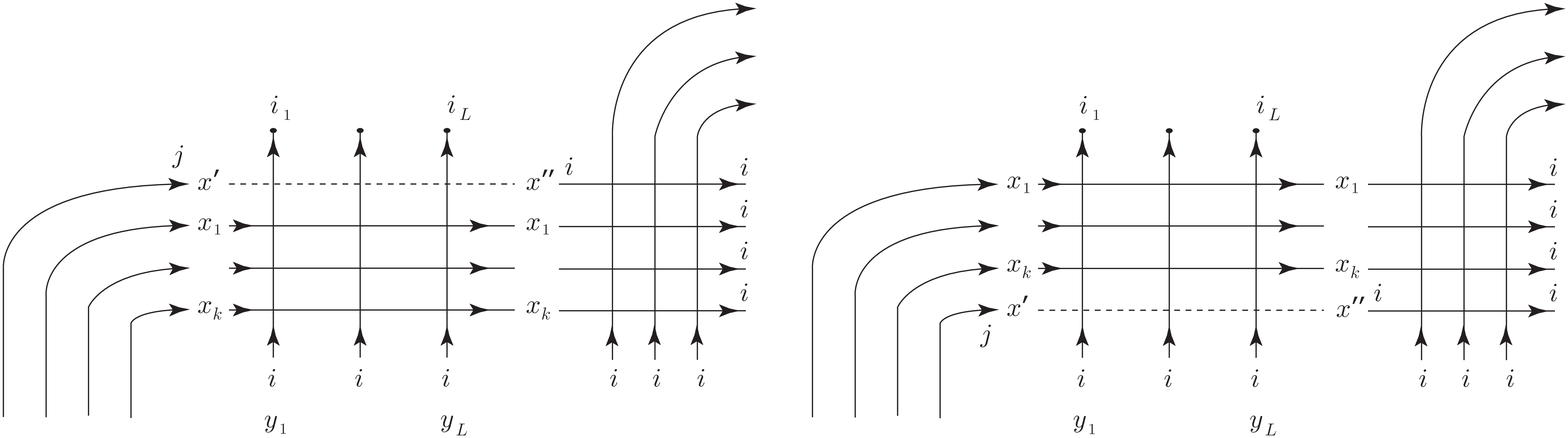}
\caption{Higgsed lattice configurations connected with other lattices.}
\label{fig:lattice_hw_3}
\end{figure}

In this case, before taking the limit $x \to \infty$, we split this variable 
into three pieces $x$, $x'$ and $x''$. 
As we will discuss in Section \ref{lattice_hanny_witten}, the splitting of 
the variable $x$ is interpreted as the lattice version of a Hanany-Witten move. 
Further, we deform the lattice configuration 
by taking the decoupling limit $x'' \to \infty$ to remove 
the horizontal line, with the variable $x''$ and colour $i$ at 
the two end boundaries, which uniquely results in no $c$ vertex 
\footnote{\,
On the gauge side, we can interpret taking the limit $x'' \to \infty$ as 
the decoupling of a pair of fundamental and anti-fundamental matter, 
with twisted masses $x''$, at the $(n+1)$-th node after the Hanany-Witten move 
in Figure \ref{fig:aM_hw} (see Section \ref{lattice_hanny_witten}).
}.

Repeating this process, one can construct the $\mathfrak{su} (M+1)$ lattice 
configurations in Appendix \ref{app:AM_v_model} 
(e.g. see Figure \ref{fig:a3_lattice}) from the lattice configuration associated 
with the simple $A_M$ linear quiver in Figure \ref{fig:am_quiver_l}.

\section{The lattice versions of Higgsing and Hanany-Witten moves}
\label{section_6}

\subsection{Korepin's specialization of parameters}

A fundamental object in exact computations in spin-chain physics is 
the domain wall 
partition function (DWPF) introduced by Korepin, in the $\mathfrak{su} (2)$
spin-$\frac12$ six-vertex model \cite{Korepin:1982gg}. 
As explained in Appendix \ref{app:A1PDWPF}, Korepin proposed a specialization of 
the parameters (the rapidities and the inhomogeneities) of the DWPF that leads to a recursion relation (and an initial condition) that completely 
determine it \cite{Korepin:1982gg}. 
In \cite{Izergin:1987}, Izergin solved Korepin's recursion relation and obtained 
a determinant expression for the DWPF (see Propositions \ref{prop:cond_su2} 
and \ref{prop:det_su2}).

\subsection{Gaiotto and Koroteev's specialization of parameters as a variation on Korepin's}

The specialization \eqref{HW_move} of the parameters, used in 
the Higgsing procedure by Gaiotto and Koroteev, is a variation on Korepin's, 
in the sense that it is essentially the same with two differences between them.

\subsubsection{The first difference}

Korepin's derivation of the recursion relation requires one type of 
conditions, and can either, while the Higgsing procedure requires both 
as in \eqref{ab_condition}. 

In the case of domain wall boundary conditions, the DWPF
is symmetric in the horizontal-line variables and also in the vertical-line variables, 
and one can choose the variables that one wishes to specialize to be those on lines 
on the boundaries of the (finite) lattice configuration on which the DWPF is defined. 
Once we do that, we have more information about the colours of the state variables 
on these lines, and only one condition is needed to derive the recursion relation.

In the case of Higgsing, when translated to the lattice, one deals with lattice 
configurations that correspond to the coordinate Bethe wavefunction which is 
not a symmetric function in the vertical-line (inhomogeneity) variables, one 
cannot (for general choices of the Higgsing parameters) associate the variables 
that one wishes to specialize to boundary lattice lines, and one requires 
(in general) two independent conditions. This makes the specialization of Gaiotto 
and Koroteev a more general version of Korepin's.

\subsubsection{The second difference}

In Izergin-Korepin-type computations, one wishes to factor out the (finite) contributions 
of a single horizontal and a single vertical lattice line (rather than completely trivialize 
them), so the parameters that are identified by Korepin's conditions are allowed to remain finite.

On the other hand, in the lattice version of Higgsing, one wishes to trivialize 
the contributions of a single horizontal and two vertical lattice lines, and 
this is achieved by identifying three parameters (using two conditions), 
then taking that parameter to infinity and normalizing appropriately. 
In the type-IIA/IIB brane realizations in Figure \ref{fig:aM_hw_brane}, 
this Higgsing procedure corresponds to Steps {\bf 1} and {\bf 2}.
This makes the specialization of Gaiotto and Koroteev a limiting case 
of Korepin's.

\subsection{The lattice version of the Hanany-Witten moves}
\label{lattice_hanny_witten}

One of the ingredients of the Higgsing of Gaiotto and Koroteev is a sequence 
of Hanany-Witten moves, in the sense that, in the type-IIA/IIB brane realizations 
in Figure \ref{fig:aM_hw_brane}, the transition from (anti-)bifundamental matter 
to new (anti-)fundamental matter is the result of Hanany-Witten moves and described 
by Step {\bf 3} \cite{Hanany:1996ie, Gaiotto:2013bwa}
\footnote{\, 
In the case of the $A_1$ quiver, there is no (anti-)bifundamental matter and no 
Hanany-Witten moves.
}.
In Figure \ref{fig:aM_hw}, 
the appearance (after Higgsing) of two pairs of fundamental and anti-fundamental 
matter from the pair of initial (before Higgsing) bifundamental and anti-bifundamental 
matter is the consequence of a sequence of Hanany-Witten moves.
On the lattice side, following the splitting of the variable $x$ in 
Figure \ref{fig:lattice_hw_3} and below, 
the disappearance of the initial horizontal-line parameter $x$, 
and the appearance of two new vertical-line parameters $x'$ and $x''$, 
which is a re-assignment of what were (initially) horizontal-line variables 
as (new) vertical-line variables, is the lattice version of the Hanany-Witten 
move.

\section{Remarks}
\label{section_7}

\subsection{Affinization}
\label{affinization}

In \cite{Bonelli:2015kpa}, Bonelli, Sciarappa, Tanzini and Vasko 
studied connections between 4D $\mathcal{N} = 2$ supersymmetric 
gauge theories and 
quantum integrable systems of the hydrodynamic type. 
In particular,
the $\widehat{A}_M$ quiver gauge theory that plays a central role, 
and appears in Figure 1, in \cite{Bonelli:2015kpa}, 
is the affine version of 
the $A_M$ quiver gauge theory that plays a central role, 
and appears in Figure 1, in the present work.
This leads us to expect that the present work has an affine extension along the 
lines of \cite{Bonelli:2015kpa}.

\subsection{Quiver with orbifold defects}

In Section \ref{sec:off_Bethe_AM}, we considered a ${\IZ}_L$ 
orbifold of the 2D gauge theory described by the $A_M$ quiver 
in Figure \ref{fig:am_quiver_l} and obtained the orbifold 
defect \eqref{defect_AM} labeled by the nested sequences 
\eqref{nest_am}. 
In \cite{Bonelli:2019jns}, Bonelli, Fasola and Tanzini studied 
a class of 4D $A_1$ quiver gauge theories, with a codimension-2 
surface defect, that supports \textit{nested instantons} 
obtained by an orbifold and labeled by nested partitions. 
They discussed a 2D gauge theory described by a quiver 
that is different from that used in the present work, and 
that corresponds to the moduli space of nested instantons. 
It would be interesting to find the relation, if any, between 
the two constructions.


\section*{Acknowledgements}
We thank
F C Alcaraz,
G Bonelli, 
A Hanany, 
H Kanno,
H C Kim,
L Piroli, 
B Pozsgay, 
A Tanzini and 
Y Terashima
for discussions and correspondence,
and the Australian Research Council for financial support.  

\appendix

\section{The $\mathfrak{su} (2)$ 
six-vertex model partial domain wall partition function}
\label{app:A1PDWPF}

\textit{We prove Proposition \ref{prop:det_su2_pDW}, to the effect 
that the partition function $\widehat{\mathcal{Z}}_{k}(\bu_k;\bm_L)$ in 
\eqref{wave_A1}, constructed from the orbifold defect 
$\widehat{\psi}_{\bI_k}^{(L)}(\bu_k;\bm_L)$ in \eqref{defect_A1}, 
agrees with the $\mathfrak{su} (2)$ partial domain wall partition function (DWPF) of the rational six-vertex model.}

\begin{lemm}\label{lemm:limit_su2}
The orbifold defect $\widehat{\psi}_{\bI_k}^{(L)}(\bu_k;\bm_L)$ 
and the partition function $\widehat{\mathcal{Z}}_k(\bu_k;\bm_L)$ 
are polynomials of degree $L-1$ in each $u_a$.
\end{lemm}

\begin{proof}
We show that $\widehat{\mathcal{Z}}_k(\bu_k;\bm_L)$ is regular at $u_A=u_B$. 
Let $\widehat{\sigma}_{A,B}$ be an operator acting on functions of 
$\bu_k$ which exchanges $u_A$ and $u_B$, and consider the defect 
$\omega_{\bI_k}(\bu_k):=\omega_{\bI_k}^{(L)}(\bu_k;\bm_L)$ in \eqref{su2_sk}.
Symmetrizing in the variables ${\bu_k}$, the pole at $u_A=u_B$, 
$A<B$, in $\omega_{\bI_k}(\bu_k)$ cancels the corresponding pole 
in $\widehat{\sigma}_{A,B}\cdot \omega_{\bI_k}(\bu_k)$. Therefore, 
$\omega_{\bI_k}(\bu_k)+\widehat{\sigma}_{A,B}\cdot \omega_{\bI_k}(\bu_k)$
has no poles at $u_A=u_B$, thus $\widehat{\psi}_{\bI_k}^{(L)}(\bu_k;\bm_L)$ and 
$\widehat{\mathcal{Z}}_k(\bu_k;\bm_L)$ are regular 
at $u_a=u_b$, $a,b=1,\ldots,L$, and polynomials of degree $L-1$ in 
each $u_a$. 
\end{proof}

\begin{lemm}
One can decouple the vector multiplet scalars by
\begin{align}
\widehat{\mathcal{Z}}_k(\bu_k;\bm_L)=
\frac{1}{L-k} \lim_{u_{k+1} \to \infty}
u_{k+1}^{1-L}\,
\widehat{\mathcal{Z}}_{k+1}(\bu_{k+1};\bm_L),
\label{limit_k_su2w}
\end{align}
and further, 
\begin{align}
\widehat{\mathcal{Z}}_k(\bu_k;\bm_L)=
\frac{1}{(L-k)!} \lim_{u_{k+1},\ldots,u_L \to \infty}
\ll \prod_{a=k+1}^L u_a^{1-L}\rr 
\widehat{\mathcal{Z}}_L(\bu_L;\bm_L)
\label{limit_su2w}
\end{align}
\end{lemm}

\begin{proof}
The orbifold defect $\widehat{\psi}_{\bI_k}^{(L)}(\bu_k;\bm_L)$ in \eqref{defect_A1} 
consists of $k!$ terms, and combining with the summation in \eqref{wave_A1}, 
$\widehat{\mathcal{Z}}_{k}(\bu_k;\bm_L)$ contains $k! \times \binom{L}{k}=L!/(L-k)!$ 
terms in the summand. Comparing the summands on the both sides of \eqref{limit_k_su2w}, we find that the leading factor of 
$\widehat{\mathcal{Z}}_{k+1}(\bu_{k+1};\bm_L)$ at $u_{k+1}\to \infty$ 
gives $(L-k)\, u_{k+1}^{L-1}\, \widehat{\mathcal{Z}}_{k}(\bu_{k};\bm_L)$, and Equations \eqref{limit_k_su2w} and \eqref{limit_su2w} are obtained.
\end{proof}

Equation \eqref{limit_su2w} states that 
by decoupling the vector multiplet scalars $u_a$, $a=k+1,\ldots,L$, 
from $\widehat{\mathcal{Z}}_L(\bu_L;\bm_L)=
\widehat{\psi}_{\bI_L}^{(L)}(\bu_L;\bm_L)$, 
one obtains $\widehat{\mathcal{Z}}_k(\bu_k;\bm_L)$. 
For $\widehat{\mathcal{Z}}_L(\bu_L;\bm_L)$, we have the following 
proposition.

\begin{prop}\label{prop:cond_su2}
The partition function $\widehat{\mathcal{Z}}_L(\bu_L;\bm_L)$ satisfies the same
four conditions that define the $\mathfrak{su} (2)$ DWPF in \cite{Korepin:1982gg}.
Namely, $\widehat{\mathcal{Z}}_L(\bu_L;\bm_L)$
\begin{itemize}
\item[{\bf 1.}] satisfies the initial condition $\widehat{\mathcal{Z}}_1(u;m)=1$, 
\item[{\bf 2.}] it is a polynomial of degree $L-1$ in each $u_a$, 
\item[{\bf 3.}] it is invariant under any permutations of $m_i$'s, and 
\item[{\bf 4.}] it satisfies the following recursion relation in $L$,
\begin{multline}
\widehat{\mathcal{Z}}_L(\bu_L;\bm_L)\Big|_{u_L=m_L-\frac{\gamma}{2}}=
\prod_{a=1}^{L-1}\ll m_L-m_a-\gamma\rr 
\ll u_a-m_L-\frac{\gamma}{2}\rr \cdot
\widehat{\mathcal{Z}}_{L-1}(\bu_{L-1};\bm_{L-1})
\label{rec_su2}
\end{multline}
\end{itemize}
\end{prop}

\begin{proof}
Condition {\bf 1} follows from the definition of $\widehat{\mathcal{Z}}_L(\bu_L;\bm_L)$, and
Condition {\bf 2} follows from Lemma \ref{lemm:limit_su2}. 
To prove Condition {\bf 3}, it is sufficient to show that 
$\omega_{\bI_L}(\bu_L)+\widehat{\sigma}_{A,A+1}\cdot 
\omega_{\bI_L}(\bu_L)$,  
in $\widehat{\mathcal{Z}}_L(\bu_L;\bm_L)$, is 
invariant under the permutation of $m_A$ and $m_{A+1}$, 
where $\omega_{\bI_L}(\bu_L)=\omega_{\bI_L}^{(L)}(\bu_L;\bm_L)$ 
is defined in \eqref{su2_sk}. 
The point is that, once this is shown, then by symmetrizing $u_A$, $u_{A+1}$, 
and $u_{A+2}$ in $\omega_{\bI_L}(\bu_L)$, 
the symmetrized function becomes invariant 
under permuting $m_A$, $m_{A+1}$, and $m_{A+2}$.
In $\omega_{\bI_L}(\bu_L)$, the factor that contains $u_A$ and $u_{A+1}$ is
\begin{multline}
\prod_{a=A, A+1}^L\ll \prod_{i=1}^{a-1}
\ll u_a-m_i-\frac{\gamma}{2}\rr 
\cdot
\prod_{i=a+1}^{L}\ll u_a-m_i+\frac{\gamma}{2}\rr \rr 
\cdot
\frac{u_{A,A+1}-\gamma}{u_{A,A+1}}
\\
\times 
\prod_{a=1}^{A-1} \frac{\ll u_{a,A}-\gamma\rr \ll u_{a,A+1} - \gamma \rr}
{u_{a,A}\, u_{a,A+1}}
\cdot
\prod_{a=A+2}^L   \frac{\ll u_{A,a}-\gamma\rr \ll u_{A+1,a} - \gamma \rr}
{u_{A,a}\, u_{A+1,a}}
\label{fac_uA}
\end{multline}
Because the factor that contains $m_A$ and $m_{A+1}$, but does not contain 
$u_A$ and $u_{A+1}$, 
\begin{equation}
\prod_{i=A, A+1}^L \ll \prod_{a=A+2}^{L}\ll u_a-m_i-\frac{\gamma}{2}\rr 
\cdot
\prod_{i=1}^{A-1}\ll u_a-m_i+\frac{\gamma}{2}\rr \rr,
\end{equation} 
is manifestly symmetric under the exchange of $m_A$ and $m_{A+1}$, 
in $\omega_{\bI_L}(\bu_L)+
\widehat{\sigma}_{A,A+1}\cdot \omega_{\bI_L}(\bu_L)$,  
it is sufficient to consider the following factor in \eqref{fac_uA}
\begin{align}
&
\ll u_{A+1}-m_A-\frac{\gamma}{2}\rr 
\ll u_A-m_{A+1}+\frac{\gamma}{2}\rr 
\frac{u_{A,A+1}-\gamma}{u_{A,A+1}}
\\
&+
\ll u_{A}-m_A-\frac{\gamma}{2}\rr 
\ll u_{A+1}-m_{A+1}+\frac{\gamma}{2}\rr 
\frac{u_{A,A+1}+\gamma}{u_{A,A+1}}
\nonumber\\
&=
-\ll m_A+m_{A+1}\rr \ll u_A+u_{A+1}\rr 
+2\ll m_Am_{A+1}+u_Au_{A+1}\rr +\frac{\gamma^{\, 2}}{2}
\nonumber
\end{align}
Since this factor is symmetric under permuting $m_A$ and $m_{A+1}$, Condition {\bf 3} is proved.
To prove Condition {\bf 4}, we assign $u_L=m_L-\frac{\gamma}{2}$ 
in $\widehat{\mathcal{Z}}_L(\bu_L;\bm_L)$, the non-zero terms only come from $\mathop{\Sym}_{\bu_{L-1}}\, \omega_{\bI_L}(\bu_L)$, and by
\begin{multline}
\widehat{\mathcal{Z}}_L(\bu_L;\bm_L)
\Big|_{u_L=m_L-\frac{\gamma}{2}}=
\\ 
\prod_{a=1}^{L-1} \ll u_L-m_a-\frac{\gamma}{2}\rr 
\ll u_a-m_L+\frac{\gamma}{2}\rr \,
\frac{u_{a,L}-\gamma}{u_{a,L}}\bigg|_{u_L=m_L-\frac{\gamma}{2}}
\times
\widehat{\mathcal{Z}}_{L-1}(\bu_{L-1};\bm_{L-1}),
\end{multline}
the recursion relation \eqref{rec_su2} is obtained.
\end{proof}

By induction in $L$, any function which satisfies the four conditions 
in Proposition \ref{prop:cond_su2} is uniquely determined
\footnote{\, 
Condition {\bf 1} gives the initial condition of the recursion, 
Condition {\bf 2} implies that the solution of the recursion is 
uniquely determined by $L$ conditions on the values of any of 
the variables $u_a$, and 
Conditions {\bf 3} and {\bf 4} give the necessary $L$ conditions 
on the values of the variable $u_L$. 
}, 
and the partition function $\widehat{\mathcal{Z}}_L(\bu_L;\bm_L)$ agrees 
with Izergin's determinant expression for the $\mathfrak{su} (2)$ DWPF
\cite{Izergin:1987} (see \cite{Wheeler:2011tm} for a review), 
which satisfies the same four conditions, 
as well as Kostov's determinant expression 
\cite{Kostov:2012jr,Kostov:2012yq}, which is equivalent to Izergin's.

\begin{prop}\label{prop:det_su2}
The partition function $\widehat{\mathcal{Z}}_L(\bu_L;\bm_L)$ is an Izergin 
determinant \cite{Izergin:1987}
\begin{multline}
\widehat{\mathcal{Z}}_L(\bu_L;\bm_L) =
\frac{\prod_{a,b=1}^L
\ll u_a-m_b-\frac{\gamma}{2}\rr 
\ll u_a-m_b+\frac{\gamma}{2}\rr }
{\prod_{a<b}^L u_{b,a}\,m_{a,b}}
\\
\times
\det \ll 
\frac{1}{\ll u_a-m_b-\frac{\gamma}{2}\rr 
\ll u_a-m_b+\frac{\gamma}{2}\rr }\rr_{a,b=1,\ldots,L},
\end{multline}
or equivalently, a Kostov determinant 
\cite{Kostov:2012jr,Kostov:2012yq},
\begin{multline}
\widehat{\mathcal{Z}}_L(\bu_L;\bm_L) =
\frac{\prod_{a,b=1}^L
\ll u_a-m_b+\frac{\gamma}{2}\rr }
{\prod_{a<b}^L u_{b,a}}
\\
\times
\det \ll 
u_a^{b-1}
\prod_{i=1}^L \frac{u_a-m_i-\frac{\gamma}{2}}{u_a-m_i+\frac{\gamma}{2}}
-(u_a-\gamma)^{b-1}\rr _{a,b=1,\ldots,L},
\end{multline}
which is an $\mathfrak{su} (2)$ DWPF, where $m_{a,b}=m_a-m_b$.
\end{prop}

\begin{proof}[Proof of Proposition \ref{prop:det_su2_pDW}]
Taking the decoupling limit \eqref{limit_su2w} in Proposition 
\ref{prop:det_su2}, Proposition \ref{prop:det_su2_pDW} is proved.
\end{proof} 

\section{The $\mathfrak{su} (M+1)$ vertex model associated with 
the $A_M$ quiver}
\label{app:AM_v_model}

\begin{figure}[t]
\centering
\includegraphics[width=120mm]{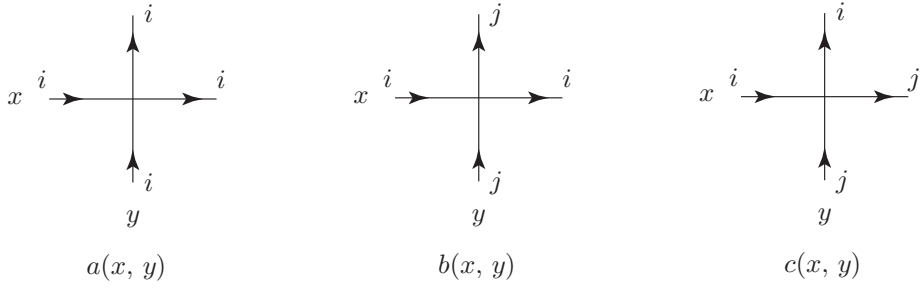}
\caption{
The non-zero vertex (Boltzmann) weights $w_v(x,y)$ 
(whose values are given in \eqref{6v_abc} in the case of the rational models, 
and in \eqref{6v_abc_tri} in the case of the trigonometric models) that can  
assigned to a vertex $v$. 
The arrows indicate the directions of rapidity and inhomogeneity parameter 
flows, and the integers $i, j=1,\ldots,M+1$, $i \ne j$, are the state variables, 
or colours assigned to the bonds.
}
\label{fig:6v_abc}
\end{figure}

\textit{We describe the lattice configurations that represent the partition 
functions of the rational $\mathfrak{su} (M+1)$ vertex model that corresponds 
to the $\mathfrak{su}(M+1)$ XXX spin-chain with spins in the fundamental 
representation.} 

\subsection{Notation}

\begin{align}
\begin{split}
&
\bx_{\bk}^M=\{\bx_{k_1}^{(1)},\ldots,\bx_{k_M}^{(M)}\}, \quad
\bx_{k_{\, p}}^{(p)}=\{x_{1}^{(p)},\ldots,x_{k_{\, p}}^{(p)}\},
\\
&
\by_{\bL}^M=\{\by_{L_1}^{(1)},\ldots,\by_{L_M}^{(M)}\}, \quad
\by_{L_{\, p}}^{(p)}=\{y_{1}^{(p)},\ldots,y_{L_{\, p}}^{(p)}\},\quad 
p=1,\ldots,M,
\end{split}
\end{align}
where $x_a^{(p)}$ and $y_a^{(p)}$ are
the rapidities and inhomogeneities, respectively. 
The translation of the lattice parameters to the gauge theory parameters 
in Table \ref{suM_matt} is 
\begin{align}
x_a^{(p)} = u_a^{(p)} - \ll \frac{M-p+1}{2} \rr \,\gamma, \qquad
y_i^{(p)} = m_i^{(p)} - \ll \frac{M-p}{2}   \rr \,\gamma, \qquad \gamma=1
\label{mass_trans_lattice}
\end{align}

\subsection{The vertex model}

A square lattice representation of the rational $\mathfrak{su} (M+1)$ vertex model 
consists of horizontal lines that carry $x$-variables that flow from right to left, 
and vertical lines that carry $y$-variables that flow from bottom to top. 
The horizontal lines and the vertical lines intersect in vertices, and each vertex 
is connected to 4 line-segments that we call \textit{bonds}.
An internal bond is connected to two vertices and a boundary bond is connected to 
a single vertex.
Each bond carries an arrow that indicates direction of the variable flow along it
\footnote{\, 
This is necessary to make the vertex type and weight, see below, of the different 
vertices unambiguous.
}. 
Further, each bond carries a colour $i \in \{1, \ldots, M+1\}$, and colour is conserved
\footnote{\,
This is the case in rational and trigonometric vertex $\mathfrak{su} (M+1)$ models,
with three types of vertices, 
but not in elliptic vertex models which has a fourth type of vertices that conserve 
colour only modulo $M+1$.
},
that is, if 
the colours on the bonds with incoming variable flows are $i$ and $j$, 
the colours on the bonds with outgoing variable flows are $k$ and $l$, then 
\begin{equation}
\label{colour_conservation}
i + j = k + l    
\end{equation}
Given colour conservation \eqref{colour_conservation}, there are three types of 
vertices, type-$a$, type-$b$ 
and type-$c$, as in Figure \ref{fig:6v_abc}, with vertex weights that depend (at 
most) on difference of variables 
\begin{align}
a(x,y)=x-y+1,\quad b(x,y)=x-y,\quad c(x,y)=1
\label{6v_abc}
\end{align}

\begin{remark}
\label{rem:trig_vertex}
The vertex weights in the trigonometric $\mathfrak{su} (M+1)$ vertex model, which 
corresponds to the $\mathfrak{su}(M+1)$ XXZ spin-chain with spins in the fundamental 
representation, are 
\begin{align}
a(x,y)=[x-y+1],\quad b(x,y)=[x-y],\quad c(x,y)=[1],
\label{6v_abc_tri}
\end{align}
where $[x]=2\sinh (x/2)$.
\end{remark}

We associate lattice configurations to the $A_M$ linear quiver in Figure \ref{fig:am_quiver}, 
where $k_{\, p} \le L_{\, p} + k_{p+1}$, $k_M \le L_M$, $p=1,\ldots,M-1$, 
as follows. For the first node with $U(k_1)$ gauge group, we associate one of the 
lattice configurations in Figure \ref{fig:quiver_lattice_1}.

\begin{figure}[H]
 \centering
  \includegraphics[width=159mm]{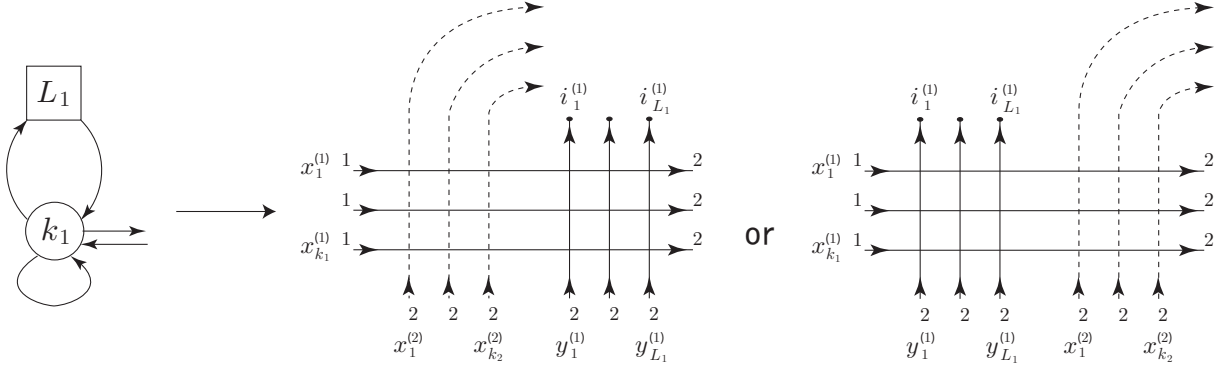}
\caption{Two possible lattice configurations that can be associated 
with the first quiver node.}
\label{fig:quiver_lattice_1}
\end{figure}

All bonds on the left boundary are assigned the (fixed and same) 
colour $1$,   
all bonds on the right and lower boundaries are assigned the (fixed 
and same) 
colour 2,
and the $L_1$ bonds on the top boundary are assigned 
(fixed but varying) colours $i_{\ell}^{(1)} \in \{1,2\}$, $\ell=1,\ldots,L_1$. 
Each dashed line that carries a rapidity variable $x_a^{(2)}$ 
is connected with another dashed line that carries a rapidity variable 
$x_a^{(2)}$ associated with the second quiver node. 
For the $p$-th quiver node with a $U(k_{\, p})$ gauge group, $p=2,\ldots,M-1$, 
we associate one of the lattice configurations in Figure \ref{fig:quiver_lattice_2}, 
where, each dashed line that carries a rapidity $x_a^{(p)}$ 
is connected with another dashed line that carries a rapidity variable $x_a^{(p)}$.

\begin{figure}[H]
 \centering
  \includegraphics[width=160mm]{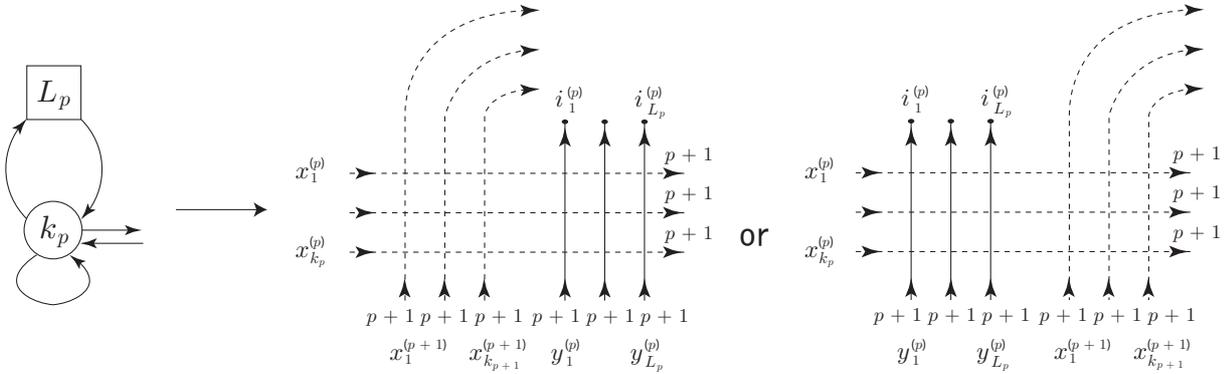}
\caption{Two possible lattice configurations that can be associated with 
the $p$-th quiver node, where $p=2,\ldots,M-1$.}
\label{fig:quiver_lattice_2}
\end{figure}

All bonds on the right boundary and on the lower boundary are 
assigned the (fixed and same) colour $p+1$, the $L_{\, p}$ bonds on the 
top boundary are assigned (fixed but varying) colours 
$i_{\ell}^{(L_{\, p})} \in \{1,\ldots,p+1\}$, $\ell=1,\ldots,L_{\, p}$, 
and we label the above left (resp. right) associated 
lattice configuration by $s_p=1$ (resp. $s_p=2$).
For the $M$-th quiver node with $U(k_M)$ gauge group, 
we associate the unique lattice configuration in Figure \ref{fig:quiver_lattice_3}.

\begin{figure}[H]
 \centering
  \includegraphics[width=95mm]{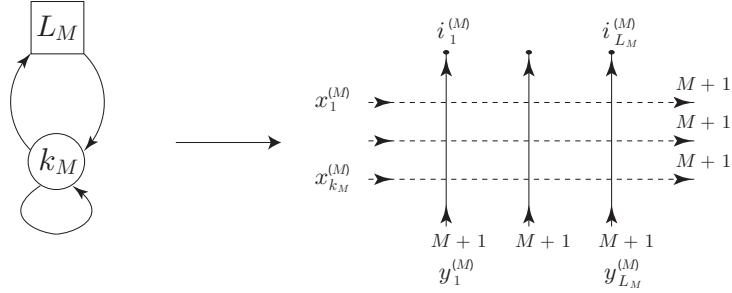}
\caption{A lattice configuration associated with the $M$-th node.}
\label{fig:quiver_lattice_3}
\end{figure}

The bonds on the right and lower boundaries are assigned the (fixed and same) colour 
$M+1$, and the $L_M$ bonds on the top boundary are assigned 
(fixed but varying) colours 
$i_{\ell}^{(L_M)} \in \{1,\ldots,M+1\}$, $\ell=1,\ldots,L_M$. 
From colour conservation, 
\begin{align}
\begin{split}
&
\sum_{p=1}^{M} \#_1 \ll
\boldsymbol{i}_{L_{\, p}}^{(p)}
\rr 
=
k_1, \quad
\sum_{p=q-1}^{M} \#_q
\ll 
\boldsymbol{i}_{L_{\, p}}^{(p)}
\rr 
=
k_q-k_{q-1},\quad
q=2, \ldots, M, 
\\ 
&
\#_{M+1}
\ll
\boldsymbol{i}_{L_M}^{(M)}
\rr =
L_M-k_M,
\end{split}
\end{align}
where $\#_q(\boldsymbol{i}_{L_{\, p}}^{(p)})$ is the number of colours $q$ in the set
$\boldsymbol{i}_{L_{\, p}}^{(p)}=\{i_1^{(p)},\ldots, i_{L_{\, p}}^{(p)}\}$. 
Examples of $\mathfrak{su} (2)$ and $\mathfrak{su} (3)$ lattice configurations 
are in Figures \ref{fig:a1_lattice} and \ref{fig:a3_lattice}, respectively.

\begin{defi}
The partition function for the rational/trigonometric 
$\mathfrak{su} (M+1)$ lattice 
configuration associated with the $A_M$ linear quiver 
in Figure \ref{fig:am_quiver} is defined by
\begin{align}
\widehat{\psi}_{\mathrm{L},\boldsymbol{i}_{\bL}^M}^{(\bS|\bL)}
\ll \bx_{\bk}^M;\by_{\bL}^M \rr 
=
\sum_{{\boldsymbol\sigma}|_{\boldsymbol{i}_{\bL}^M}} \prod_{v \in \bv|_{\boldsymbol\sigma}} 
w_v
\ll 
\mathsf{x}_1^{(v)}, \mathsf{x}_2^{(v)} 
\rr,
\label{am_lattice_pf}
\end{align}
where $\boldsymbol{i}_{\bL}^M=
\{\boldsymbol{i}_{L_1}^{(1)},\ldots,\boldsymbol{i}_{L_M}^{(M)}\}$ 
and $\bS=\{s_1,\ldots,s_{M-1}\}$, $s_p \in \{1,2\}$. 
${\boldsymbol\sigma}|_{\boldsymbol{i}_{\bL}^M}$ is a lattice 
configuration with fixed boundary colours $\boldsymbol{i}_{\bL}^M$, 
and $\bv|_{\boldsymbol\sigma}$ is the set of all vertices on 
the lattice with the lattice configuration ${\boldsymbol\sigma}$. 
$w_v(\mathsf{x}_1^{(v)}, \mathsf{x}_2^{(v)})$ is the vertex weight 
on the vertex $v$ defined in Figure \ref{fig:6v_abc}, where 
$\mathsf{x}_1^{(v)}$ and $\mathsf{x}_2^{(v)}$ are the corresponding 
variables. 
\end{defi}


\end{document}